\documentclass[11pt,fleqn]{article}
\usepackage{eurosym}
\usepackage{amsfonts,amsmath,amssymb,graphicx,setspace}
\usepackage[usenames, dvipsnames]{xcolor}
\usepackage[colorlinks=true,linkcolor=blue,citecolor=Mahogany,urlcolor=blue]%
{hyperref}
\usepackage[top=1in, bottom=1in, left=1in, right=1in]{geometry}
\usepackage[normalem]{ulem}
\usepackage{subfig}
\usepackage{amsmath}
\usepackage{amsfonts}
\usepackage{amssymb}
\usepackage{amsthm}
\usepackage{mathtools}
\usepackage{graphicx}
\usepackage[longnamesfirst]{natbib}
\usepackage[normalem]{ulem}
\definecolor{accentcolor}{RGB}{26, 148, 49}
\usepackage{hyperref}
\hypersetup{pdfborderstyle={},allcolors=accentcolor,colorlinks}

\usepackage{titling}
\usepackage{bbm}
\usepackage{apptools}

\usepackage{etoolbox}
\newcommand{\zerodisplayskips}{
  \setlength{\abovedisplayskip}{8pt}
  \setlength{\belowdisplayskip}{7pt}
  \setlength{\abovedisplayshortskip}{8pt}
  \setlength{\belowdisplayshortskip}{7pt}}
\appto{\normalsize}{\zerodisplayskips}
\appto{\small}{\zerodisplayskips}
\appto{\footnotesize}{\zerodisplayskips}

\AtAppendix{\counterwithin{lemma}{section}}
\AtAppendix{\counterwithin{claim}{section}}
\usepackage[
]{ifdraft}

\usepackage[font=normalsize,justification=centering]{caption}
\DeclareCaptionLabelFormat{captionf}{\textsc{#1~#2}}
\usepackage{lscape}
\usepackage{placeins}
\providecommand{\U}[1]{\protect\rule{.1in}{.1in}}
\usepackage[noamsmath]{kpfonts}
\usepackage{inconsolata}

\usepackage[T1]{fontenc}
\providecommand{\U}[1]{\protect\rule{.1in}{.1in}}

\newtheorem{proposition}{Proposition} 
\newtheorem{corollary}{Corollary}

\newtheorem{claim}{Claim}

\begin{document}

\title{Disentangling Exploration from Exploitation\thanks{We thank Arjada Bardhi, Matthew Ellman, Francesco Fabbri, Nicolas Klein, Xiaosheng Mu, and Bruno Strulovici for helpful comments. We gratefully acknowledge financial support from the National Science Foundation, grant SES 1949381.}} \author{Alessandro Lizzeri\thanks{Princeton University and NBER; \href{mailto:lizzeri@princeton.edu}{\texttt{lizzeri@princeton.edu} }} \hspace{3mm} \hspace{2mm} Eran Shmaya\thanks{State University of New York at Stony Brook; \href{mailto:eran.shmaya@stonybrook.edu}{\texttt{eran.shmaya@stonybrook.edu}}} \hspace{3mm}\hspace{2mm} Leeat Yariv\thanks{Princeton University, CEPR, and NBER; \href{mailto:lyariv@princeton.edu}{\texttt{lyariv@princeton.edu}}}} \date{\today }

\pretitle{\begin{flushleft}\LARGE} 
\posttitle{\end{flushleft}}
\preauthor{\begin{flushleft}\large} 
\postauthor{\end{flushleft}}
\predate{\begin{flushleft}} 
\postdate{\end{flushleft}}
\settowidth{\thanksmarkwidth}{*}
\setlength{\thanksmargin}{-\thanksmarkwidth}

\maketitle \thispagestyle{empty}

\renewenvironment{abstract} {\par\noindent\textbf{\abstractname.}\ \ignorespaces} {\par\medskip} \begin{abstract} Starting from \cite{robbins1952some}, the literature on experimentation via multi-armed bandits has wed exploration and exploitation. Nonetheless, in many applications, agents' exploration and exploitation need not be intertwined: a policymaker may assess new policies different than the status quo; an investor may evaluate projects outside her portfolio. We characterize the optimal experimentation policy when exploration and exploitation are disentangled in the case of Poisson bandits, allowing for general news structures. The optimal policy features complete learning asymptotically, exhibits lots of persistence, but cannot be identified by an index à la Gittins. Disentanglement is particularly valuable for intermediate parameter values.

\vspace{4mm}

 \vspace{2mm}


\end{abstract}
\newpage \pagenumbering{arabic}

\section{Introduction}

In various applications, decision-makers navigate a dynamic landscape by simultaneously taking actions and gathering insights about their environment. Policymakers evaluate the performance of new policies while managing existing ones. Investors assess their financial portfolios, gauging immediate returns and future prospects. Employees navigate their career paths by exploring opportunities within their organization or beyond. 

The seminal work of \cite{robbins1952some}, \cite{gittins1979bandit}, and \cite{gittins1979dynamic}, proposed a dynamic model that fuses learning with decision-making. In their classical multi-armed bandit problem, each action taken by an agent provides insights solely into that specific action’s effectiveness. Like a bet on a slot machine---where one must pay to learn the outcome---optimal choices balance the benefits of learning about the action (exploration) and its consequent payoff benefits (exploitation).

We propose a framework for studying settings in which exploration and exploitation are, and can be, untangled. We leave behind the slot machine model, and instead consider decision-makers who can learn about choices they might not immediately pursue. We characterize the resulting optimal policy and illustrate when the ability to disentangle exploration from exploitation is especially advantageous.

In our model, an agent encounters a recurring decision between two uncertain projects. These projects might be policies, stocks, job prospects, etc. To simplify, we assume that each project offers either a positive flow payoff if successful (good project) or no payoff if unsuccessful (bad project). The quality of each project is determined independently at the outset, with prior probabilities known to the agent.

In each period, the agent decides which project to exploit; that is, which policy to implement, which investment to make, which job to choose, etc. Her choices determine the overall payoff, calculated as the discounted sum of rewards obtained from exploitation. The agent learns incrementally throughout the process: at the start of each period, she possesses a unit of attention, or exploration, which she allocates between the two projects. When exploring a project, the agent can get conclusive information about its quality, which arrives at a Poisson rate. The arrival rate may differ depending on the explored project and whether it is good or bad.

In contrast to the traditional multi-armed bandit framework, our model's operating assumption is that the agent can gather information through exploration, but not through exploitation. However, in many real-world scenarios, the exploited project yields some valuable data. To accommodate this, we introduce a constrained version of our model where a predetermined portion of exploration is allocated to the exploited project each period. When the predetermined portion is set to $0$, exploration and exploitation are entirely disentangled. Conversely, when this portion is set to $1$, exploration and exploitation are fully entangled and our environment admits several settings studied in the literature if we assume further that one project is known to be good. Specifically, when news arrives only about good projects at positive rates, this aligns with the \cite{keller2005strategic} (KRC) setting. When news arrives exclusively about bad projects at positive rates, this mirrors the \cite{keller2010strategic} (KR) setting.\footnote{While special, these settings have been used to study a variety of applications, including delegation problems \citep{horner2013incentives,guo2016dynamic}, experimentation by committee \citep{strulovici2010learning}, dynamics of discrimination \citep{bardhi2020early}, and many others; see our literature review.}

We first show that, whenever some portion of exploration can be dedicated to an unexploited project, an optimizing agent exploits the \textit{realized} best project asymptotically. Intuitively, if there is any room for the agent to be swayed by information toward exploiting a different project than the one she already exploits, any level of disentanglement would allow her to gain that information in the long run. The asymptotic optimality in our setting underscores a fundamental difference from the conventional setting, with full entanglement, where it is well-known that the agent’s exploitation need not converge to the ex-post optimal project.

We start with the special case in which one project is known to be good, and therefore safe, as in KRC and KR, although we allow for arbitrary Poisson arrival rates of news. In this case, the agent explores the uncertain, or risky, project as much as possible. With any level of entanglement, the agent’s exploitation choices constrain her exploration. She thus faces the standard exploration / exploitation dilemma. We show that the optimal strategy involves setting a threshold on the posterior probability of the risky project’s favorability. When this threshold is surpassed, the agent chooses to exploit and further explore the risky project; otherwise, she opts for the safe project, minimally exploring it using the predetermined portion of her attention budget.

In Proposition \ref{prop1}, we demonstrate that the optimal threshold depends only on the maximum between the arrival rates of good and bad news. In general good news settings, where good news arrives faster than bad news, receiving no news makes the agent increasingly pessimistic. In general bad news settings, where bad news arrives more rapidly, receiving no news makes the agent increasingly optimistic. In either case, the analytical description of the optimal threshold is identical. The optimal policy exhibits different features, naturally. In particular, as in KRC and KR, with a high enough initial prior that the risky project is good, absent news arrival, the agent ultimately switches her exploitation in good news settings, but never does so in bad news settings.

The optimal policy changes when the returns to both projects are uncertain. To illuminate the forces within our model, we focus on scenarios where exploration and exploitation are entirely disentangled. Then, the agent optimizes exploitation by favoring the myopically optimal project at any given moment. However, determining the optimal exploration strategy is less straightforward and does not adhere to an index policy akin to Gittins’. 

We begin by examining balanced news settings, where both good and bad news arrive at equal rates for each project. In such settings, the passage of time without any news does not provide any insight into the quality of a project. Consequently, the optimal policy remains constant, and the primary consideration is which project to explore at the outset.

In Proposition \ref{prop2}, we show that the optimal exploration strategy is determined via the comparison of a particular formulation of the information value associated with each project. This value is influenced not only by the rates at which news arrives but also by the relative rewards and prior probabilities assigned to each project’s success. Specifically, when one project is significantly more likely to succeed compared to another, the inferior project may hold greater information value, as there is a higher probability that new information could lead the agent to switch the exploited project. Our characterization highlights two key deviations from the optimal policy observed in the traditional, fully entangled environment. First, an increase in the prior probability of a project’s success may prompt the agent to explore the alternative project in our environment, but not in the classical environment. Second, the optimal exploration strategy is intricately linked to the interplay between the parameters of both projects and, as noted, cannot be described via a separable index.

In general good news settings, Proposition \ref{prop3} shows that the agent still optimally exhibits a lot of persistence in her exploration. Absent news, the agent switches which project she explores at most once. This switch occurs only if the initially explored project aligns with the myopically optimal one, i.e., the project promising higher expected payoffs.

This outcome is rooted in a fundamental principle of information economics: valuable information is actionable and influences which project is exploited. In general, actionable information manifests in two forms, either adverse news regarding the exploited project, or favorable news concerning the alternative project. To glean intuition for the persistence of optimal exploration, consider pure good news settings, where only good news arrives at a positive rate about either project, as in KRC. In such settings, in the short run, actionable information materializes only through positive news about the unexploited project. If the agent explores the unexploited project, absent news, she becomes increasingly pessimistic about the explored project. Consequently, she has no incentive to switch either her exploited project or her explored project. 

The optimality of persistent exploration starkly contrasts with predictions derived from the classical, fully entangled environment. In the classical good news setting, as the agent explores and exploits a project, her confidence in its potential diminishes gradually, leading to a reduction in its corresponding Gittins index. Eventually, the indices for both projects align, prompting the agent to alternate between the projects until more information emerges, hence switching infinitely often. Subsequently, upon receiving positive news about either project, the agent indefinitely explores and exploits that project, effectively terminating further information gathering. In particular, with some probability, the agent ultimately exploits the project deemed inferior ex-post.

In general bad news settings, Proposition \ref{prop4} illustrates that optimal exploration strongly depends on projects’ potential rewards. Once the agent embarks on exploring the high-reward project, she remains committed to it without changing her exploration unless information arrives. Furthermore, in the absence of news, the agent inevitably explores the high-reward project at some point. Thus, similar to the dynamics observed in good news settings, the agent may switch her exploration at most once without news arrival. 

To gain intuition, consider pure bad news settings, where only bad news arrives at positive rates, as in KR. In such settings, when the agent explores the high-reward project and no news is received, her confidence in the project progressively grows. Only negative news regarding that project would prompt her to switch her exploited project. Therefore, it remains optimal to continue exploring the high-reward project. One might question why the same logic wouldn't apply to the low-reward project. For the low-reward project, even if the agent maintains a sufficiently optimistic outlook, positive information about the high-reward project could still sway her exploitation choice. The only means of acquiring such information is by exploring the high-reward project for an extended period.

The distinction from the classical environment hinges on the nature of news arrival. In good news settings, the separation of exploration from exploitation results in a higher level of persistence in the optimal policy. Conversely, in bad news settings, there tends to be comparatively less persistence. Indeed, in the classical bad news environment, once the agent initiates exploration and exploitation of a project, the absence of news fosters a growing optimism towards the project. Thus, regardless of the project's potential reward, the agent optimally refrains from switching to an alternative.

In the settings we consider, the payoff benefits of disentanglement are most pronounced when parameters fall within intermediate ranges: the discount rate, arrival rates of news, and initial beliefs regarding the viability of the projects under consideration. Collectively, our results show that when information and actions occur in sync, the ability to disentangle the two not only impacts behavioral predictions, but carries important implications for potential payoffs. 

\section{Related Literature}

The multi-armed bandit problem was likely initially posed by \cite{thompson1933likelihood} in the context of clinical trials. Starting from \cite{robbins1952some}, the statistics literature has offered insights on the features of optimal policies. \cite{gittins1979bandit} and \cite{gittins1979dynamic} present the first general index-based optimal policies. \cite{gittins2011bandits} offer a survey of ensuing results. As already noted, the special case of Poisson bandits was introduced by \cite{keller2005strategic} (KRC) and \cite{keller2010strategic} (KR), assuming two arms, only one of which yields uncertain rewards. 

The basic multi-armed bandit setting has been utilized for a wide array of applications in economics, ranging from monopoly pricing decisions \citep{rothschild1974two}, to labor market choices and matching \citep{jovanovic1979job, miller1984job}, to venture capital \citep{bergemann1998venture}, to the design of recommender systems \citep{che2018recommender}, to team experimentation \citep[][in adition to KRC and KR]{bolton1999strategic, strulovici2010learning}; for a survey, see \cite{bergemann2006bandit}.\footnote{The analysis in \cite{che2018recommender} relates to the special case of one safe project in our environment, which we discuss in Section \ref{SafeProject}. \cite{eliaz2022optimal} consider an extension of the basic model, where bandits---or tasks, in their framework---evolve when attended to, and payoffs also depend on unattended tasks. There is also recent empirical work that uses the basic multi-armed bandit setting in the context of pharmaceutical demand and physician prescribing behavior \citep[see ][]{crawford2005uncertainty, currie2020understanding,dickstein2021efficient} and in the context of research and development \citep[][]{zhuo2023exploit}.} 

Our paper also relates to the literature on dynamic information acquisition, initiated by \cite{wald1947foundations}. In the most basic model, an agent can acquire costly signals in sequence, and determine when to stop information collection and take a decision. In our setting, the cost of exploring one project is the option value of exploring the other. Unlike the classical model, the cost is therefore changing and endogenous. Furthermore, while our setting is dynamic, it does not correspond to a stopping problem per se.\footnote{\cite{damiano2020learning} study the KRC setting where an agent can also acquire costly auxiliary information, disconnected from exploitation, which produces conclusive news at Poisson rates. They show ways by which the information optimally acquired depends on the agent's posterior.} 

The idea that decision makers may be able to attend to or acquire information only up to a limit also appears in the rational inattention literature, see \cite{sims2003implications} and \cite{mackowiak2023rational}'s survey. Recent work considers dynamic attention allocation. For example, \cite{che2019optimal} consider an environment \'a la \cite{wald1947foundations}---a stopping problem---in which a decision maker acquires information from different news sources, each providing conclusive news about the underlying state at a Poisson rate, prior to making an irreversible binary decision. Since the rates at which news arrives from either source may depend on the underlying state, the optimal policy balances the speed at which either news source delivers news and its ``bias,'' a trade-off different than the one underlying our agent's problem. \cite{liang2022dynamically} also study a variation of the Wald problem, where a decision maker allocates a fixed attention budget across multiple sources of information to learn about a decision-relevant state. Information sources are diffusion processes whose unknown drift is an attribute that contributes linearly to determine the state. In the optimal policy, the decision maker initially allocates all attention to the most informative source, then gradually incorporates additional sources until, eventually, attends to all sources.  

There is also a literature in computer science that takes an algorithmic approach to identifying which arm is most desirable in a multi-armed bandit problem. \cite{bubeck2011pure} is perhaps the most conceptually related to our paper. They focus on regret-minimizing exploration algorithms. There is no simultaneous exploitation, and the objective is the difference between the average payoff of the best arm and the average payoff obtained by the algorithm's recommendation. See also \cite{audibert2010best} and the literature that followed.

\section{The Model}

An agent allocates exploration and exploitation resources between two projects, $L$ and $H$, in continuous time. Project $z=L,H$ is good with probability $p_{z}$ and bad with the complementary probability $1-p_{z}$. The quality of the two projects is determined independently. If project $z$ is good, it pays a flow reward of $R_{z}>0$; If project $z$ is bad, it pays $0$ forever. We assume $R_{H}>R_{L}>0$. We also assume that $p_{L}, p_{H}>0$ and $p_{H}<1$ so that there is meaningful uncertainty about which project is superior.

As in KRC, we assume that the agent has a unit of investment to allocate, capturing the exploitation aspect of the agent's choice. At any moment, the agent's instantaneous reward from investing $k_{z}\geq 0$ in exploiting project $z=L,H$ is given by: 
\begin{equation*} 
k_{L}\tilde{R}_{L}+k_{H}\tilde{R}_{H}, 
\end{equation*} 
where $k_{L}+k_{H}=1$ and $\tilde{R}_{z}$ denotes the realized rewards from project $x=L,H$. We assume the agent's exploitation policy is measurable with respect to the information available at any time. As is standard, payoffs are discounted at a fixed rate $r>0$.\footnote{We later show that for most of our analysis, the agent optimally chooses $k_{z}\in \{0,1\}$ for $z=L,H$. We maintain this greater generality in order to contrast some of our results with the classical, fully entangled setting, where interior investments are sometimes utilized in the optimal policy.}

Analogously, at any moment, the agent allocates a unit budget of attention, or information collection resources, across the projects. This is the exploration aspect of the agent's choice. If the agent spends a fraction $\alpha _{z}>0$ of her attention budget exploring project $z=L,H$, she may receive conclusive news about project $z$. Specifically, if project $z$ is good, the agent receives good news---a conclusive signal indicating that the project is good---at a Poisson rate $\alpha _{z}\lambda_{z}^{g}$ (and no news otherwise). Similarly, if project $z$ is bad, the agent receives bad news---a conclusive signal asserting the project is bad---at a Poisson rate $\alpha _{z}\lambda_{z}^{b}$. We assume $\max \{\lambda_{z}^{g},\lambda_{z}^{b}\}>0$ and that $sign(\lambda_{H}^{g}-\lambda_{H}^{b})=sign(\lambda_{L}^{g}-\lambda_{L}^{b})$, with the convention that $sign(0)=0$. That is, the agent has opportunities to learn and the information structure is similar across the two projects. As for exploitation, we assume the agent's exploration policy is measurable with respect to the information available at any time.

Whenever $\lambda_{z}^{g}-\lambda_{z}^{b}>0$ for $z=L,H$, good news arrives at a higher rate than bad news. We refer to such environments as \textit{good news settings}. Absent any news, the agent becomes increasingly pessimistic: no news is bad news. A special case corresponds to the frequently studied good news setting of KRC, which we term \textit{pure good news}, where $\lambda_{z}^{g}>0$ and $\lambda_{z}^{b}=0$ for $z=L,H$. Conversely, whenever $\lambda_{z}^{b}-\lambda_{z}^{g}>0$ for $z=L,H$, bad news arrives at a higher rate than good news. We refer to such settings as \textit{bad news settings.} Absent any news, the agent becomes increasingly optimistic: no news is good news. A special case corresponds to the frequently studied bad news setting of KR, which we term \textit{pure bad news}, where $\lambda_{z}^{b}>0$ and $\lambda_{z}^{g}=0$ for $z=L,H$. We refer to settings in which good and bad news arrive at precisely identical rates, $\lambda_{z}^{g}=\lambda_{z}^{b}$ for $z=L,H$, as \textit{balanced news settings}. In balanced news settings, without the arrival of news, the agent's posterior belief that the explored project is good does not change. These settings are useful as central reference cases around which we construct some of our proofs.

We assume that payoffs, which depend only on exploitation choices, are unobserved throughout the decision-making process. This assumption is a natural benchmark in pursuit of our goal of understanding the consequences of disentangling information collection from payoff-relevant actions. The assumption is also a reasonable approximation in a number of applications. For instance, the consequences of particular policy choices may become apparent only in the fullness of time.\footnote{Indeed, in the U.S., the Council of Economic Advisers is charged with ``advising the President on economic policy based on data, research, and evidence''; see \url{https://www.whitehouse.gov/cea/}.} Similarly, returns to long-run financial investments---like retirement savings---may provide weak signals regarding the future promise of underlying stocks, and investors may explore features of a variety of stocks, independent of their portfolio. Financial investment in charitable causes also frequently provides limited information on the charities' value.\footnote{In fact, there are various resources designed to assist donors explore various charities; see, e.g., \url{https://www.charitynavigator.org/}.} Last, employees can certainly observe their wages, but absent explicit queries, may not learn about their future prospects in their place of employment. Furthermore, employees can explore opportunities in their existing job, or elsewhere.

Certainly, in many applications, rewards from exploitation choices do provide some information about the quality of the undertaken projects. In order to capture such environments, as well as relate the commonly utilized exploration/exploitation model to ours, we consider the $\alpha $\textit{-constrained decision process}. In the $\alpha $-constrained decision process, whenever the agent exploits project $z=L,H$, she must allocate at least $\alpha$ to exploring it: $\alpha_{z}\geq \alpha $. When $\alpha =1$, the agent must explore the project she exploits, corresponding to the standard exploration/exploitation trade-off. When $\alpha =0$, exploration and exploitation are fully disentangled. 

Throughout, we characterize optimal policies up to measure-0 sets of time.

\bigskip

We begin with a straightforward result that highlights the fact that the option to disentangle exploration from exploitation, corresponding to any $\alpha <1$, has important implications on outcomes.

\bigskip

\begin{proposition}[Asymptotic Optimality]\label{prop0} 
For all $\alpha<1$, the agent exploits the best project asymptotically.
\end{proposition}

\bigskip

Proposition \ref{prop0} offers a fundamental contrast between our environment and the standard setup, where it is well known that the agent's exploitation need not converge to the ex-post optimal project.

The proof of Proposition \ref{prop0} holds for any number of projects and any payoff process. To prove this result, we need to show that the agent will eventually explore projects for a sufficiently long time so as to learn to exploit the best one. Now, an impatient agent might prefer an exploration strategy that is more efficient in the short run. Assume, for instance, that $p_{L}$ is close to $1$, while $p_{H}$, $\lambda_H^g$, and $\lambda_H^b$ are low. In the long run, the agent benefits from exploring project $H$. In the short run, however, exploring project $H$ is not useful since, in expectation, it would take a long time to conclude that project $H$ is good with sufficient likelihood to exploit it. In fact, in the classical environment, if project $L$ is known to be good, a sufficiently impatient agent would never learn that project $H$ is good as well. With $\alpha<1$, the impatient agent may still explore project $L$ initially: if $\lambda_L^b$ is sufficiently high, the agent might initially explore project $L$ since bad news will lead her to switch her exploited project. However, as we show in the proof, at some point, the short-run benefit from continuing to explore project $L$ diminishes enough so that even an impatient agent will prefer to explore project $H$.

Proposition \ref{prop0} also underscores the importance of our assumption that the agent is long-lived. If we replace our agent with a sequence of short-lived agents, each of whom lives for a fixed duration, then it may be that they all prefer to explore project $L$ since neither will stick around long enough to benefit from exploring project $H$. \cite{liang2020complementary} call this phenomenon a \emph{learning trap}.

\section{One Safe Project\label{SafeProject}}

As already noted, a heavily studied exploration/exploitation setting is that introduced by KRC, where project $L$ is ``safe:'' $p_{L}=1$. This setting is used in many applications and is a special case of our environment.

\subsection{Optimal Policy with a Safe Project}

With one safe project, optimal exploration is trivial. Because uncertainty is present only for project $H$, the agent explores project $H$ as much as she can (with at least $1-\alpha$ units of attention).\footnote{The results are the same if there is an exogenous baseline arrival rate of news on the risky project that is independent of the exploited project, where the exploitation decision generates additional information.} The choice of exploitation is less obvious. A myopic agent would exploit project $H$ when it has a higher expected value, whenever her posterior that project $H$ is good exceeds $p_{M}=\frac{R_{L}}{R_{H}}$. With $\alpha=0$, the agent exploits project $H$ only when it is myopically optimal, namely when $p_H \geq p_M$. When $\alpha >0$, exploiting project $H$ garners an informational advantage as it allows the agent to explore project $H$ and learn at higher rates: she can dedicate her full attention to project $H$ instead of only a fraction $1-\alpha$ of it. The agent may then exploit project $H$ at even lower posteriors than $p_{M}$, an instance of the exploration/exploitation trade-off. The following proposition characterizes the optimal exploitation strategy.\footnote{The result is essentially implied by a combination of results in \cite{che2018recommender}, although they study a different set of questions. Our method of proof is different and, we believe, instructive.}

\bigskip

\begin{proposition}[One Safe Project: Optimal Exploitation]\label{prop1} 
Let $\lambda =\max \{\lambda_{H}^{g},\lambda_{H}^{b}\}$. For any $\alpha \in \lbrack 0,1]$, the agent optimally exploits project $H$ whenever her posterior that project $H$ is good exceeds $\bar{p}\left( \alpha \right)$, where \begin{equation*} \bar{p}\left( \alpha \right) =\frac{\left( r+\lambda \left( 1-\alpha \right) \right) R_{L}}{(r+\lambda )R_{H}-\lambda \alpha R_{L}}. \end{equation*} The cutoff $\bar{p}(\alpha )\leq \frac{R_{L}}{R_{H}}$ is decreasing in $\alpha$ and $R_{H}/R_{L}$, and increasing in $r$. When $\alpha >0$, it is decreasing in $\lambda$. 
\end{proposition}

\bigskip

Although the cutoff $\bar{p}\left( \alpha \right) $ does not depend on whether good news or bad news arrive at higher rate, provided the maximal news arrival rate $\lambda $ remains constant, the optimal policy differs between the two settings. In good news settings, if no news arrives, any amount of exploration of project $H$ leads the agent to grow increasingly pessimistic about project $H$. If the agent starts by exploiting project $L$, she switches to exploiting project $H$ only upon receiving good news. If the agent starts by exploiting project $H$, after a sufficiently long time without news, the agent becomes sufficiently pessimistic about that project that she switches to exploiting project $L$. In contrast, in bad news settings, if no news arrives, any amount of exploration of project $H$ leads the agent to grow increasingly optimistic about project $H$. Therefore, if the agent starts by exploiting project $L$, absent bad news, she switches to exploiting project $H$ at some point. If she starts by exploiting project $H$, she never switches unless bad news arrives.

The KRC and KR cutoffs correspond to $\bar{p}\left( 1\right) $. As $\alpha$ decreases, the link between exploration and exploitation is relaxed and $\bar{p}\left( \alpha \right)$ approaches the myopic cutoff $p_{M}$. When $\frac{R_{H}}{R_{L}}$ increases, gaining information on whether project $H$ is good becomes more valuable and the cutoff $\bar{p}\left( \alpha \right)$ moves away from $p_{M}$. Last, as $\lambda$ increases, exploration of project $H$ becomes more appealing as it is expected to yield a conclusive signal more quickly. Again, the optimal cutoff $\bar{p}\left( \alpha \right)$ moves away from $p_{M}$.

In order to glean intuition for the derivation of the optimal cutoff, consider a good news setting. For any posterior $p$ such that $pR_{H}\geq R_{L}$, it is certainly optimal for the agent to exploit project $H$: it generates higher expected payoffs and delivers more information. Assume then that $pR_{H}<R_{L}$. Call $\sigma _{L}$ the strategy that specifies exploiting project $L$ until news, and $\sigma_{\Delta}$ an alternative strategy that prescribes exploiting project $H$ for a short time interval $\Delta>0$ before returning to exploiting project $L$ in the event that there is no news. The difference in payoffs between these two strategies is given by: 
\begin{equation} 
-\Delta r\left( R_{L}-pR_{H}\right) +\left( 1-\Delta r\right) p\lambda\Delta \alpha \frac{r}{r+\left( 1-\alpha \right) \lambda }\left(R_{H}-R_{L}\right) +O(\Delta^2). \label{GNDeviation} 
\end{equation}

The first term in equation (\ref{GNDeviation}) is the expected flow payoff difference between exploiting projects $L$ and $H$. The second term is the expected discounted present value of information that reflects the possibility that, in the time interval $\Delta$, the agent receives good news and optimally switches to exploiting project $H$. The arrival rate of bad news appears only in a term corresponding to the discounted flow payoff during the interval of length $\Delta$ if bad news is received from project $H$ (the agent intends to switch back to project $L$ absent news). Since the probability of such news, when project $H$ is bad, is $O(\Delta)$, the corresponding term is $O(\Delta^2)$. At the cutoff $\bar{p}\left( \alpha \right) $, taking limits as $\Delta \rightarrow 0$, our proof illustrates that the expression in equation (\ref{GNDeviation}) approaches $0$. This yields the formula appearing in Proposition \ref{prop1}. 

An analogous construction holds for bad news settings, and the resulting cutoff depends on the maximal arrival rate for both good news and bad news settings. In particular, the cutoff corresponding to $\lambda_H^i>\lambda_H^j$, where $i,j \in \{g,b\}$ is the same as the cutoff corresponding to a setting with $\lambda_H^i$ and $\lambda_H^j=\lambda_H^i - \epsilon$, with $\epsilon>0$ as small as desired. It follows that the cutoff corresponding to a good news setting with good news arriving at a rate of $\lambda$ is the same as the cutoff for a balanced news setting with arrival rate of $\lambda$. Similarly, the cutoff corresponding to a bad news setting with bad news arriving at a rate of $\lambda$ is also the same as the cutoff for a balanced news setting with arrival rate of $\lambda$. Thus, the cutoff formulas for both good and bad news settings must coincide.

\subsection{Payoff Consequences of Disentanglement}

Relaxing the entanglement constraint by reducing $\alpha$ can only improve the agent's expected payoff. We now identify features of the environment that make disentanglement particularly valuable.

Certainly, when $R_{H}/R_{L}$ increases, the benefits of learning without forgoing payoffs are larger. Therefore, the value of disentanglement increases in $R_{H}/R_{L}$. In what follows, we inspect the dependence of payoffs on other parameters.

For any project rewards $R_{L}$ and $R_{H}$, denote by $\Pi (p_{H},r/\lambda;\alpha )$ the agent's expected payoff for the environment's parameters, an analytical formulation of which appears in the Appendix. To quantify the impacts of disentanglement, we focus on the two extreme cases, $\alpha =0$ and $\alpha =1$, and consider the normalized payoff difference: 
\begin{equation*} 
\Delta \Pi (p_{H},r/\lambda )=\frac{\Pi (p_{H},r/\lambda ;0)-\Pi(p_{H},r/\lambda ;1)}{p_{H}R_{H}+(1-p_{H})R_{L}}, 
\end{equation*} 
where the denominator represents the ex-ante value of the full information payoff and serves as a natural normalization factor. In Figure \ref{Figure1}, we depict $\Delta\Pi (p_{H},r/\lambda )$ for various parameters, focusing on the pure good and bad news settings, where $\lambda_{H}=\max \{\lambda_{H}^{g},\lambda_{H}^{b}\}$ and $0=\min \{\lambda_{H}^{g},\lambda_{H}^{b}\}$.

\begin{figure}[t]
\includegraphics[width=1\columnwidth]{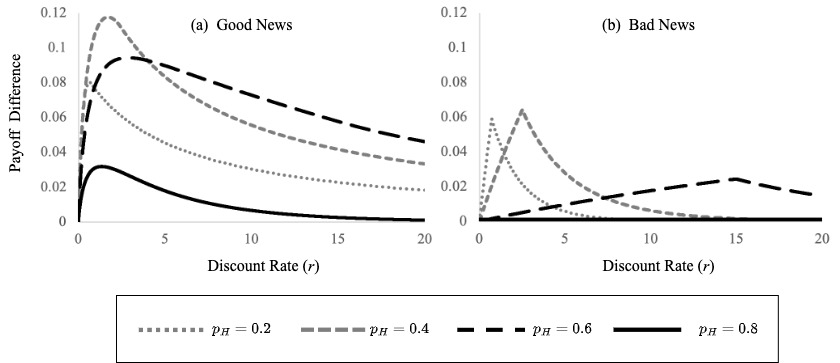}
\vspace{1mm}
\caption{Payoff value of disentanglement for (a) pure good news settings, and (b) pure bad news settings when $R_{L}=10$, $R_{H}=15$, and $\lambda_{H}=5$\label{Figure1}}
\end{figure}

As can be seen, the benefit of disentanglement is non-monotonic with respect to the discount rate $r$, and equivalently, with respect to the arrival rate $\lambda$ of good news. Intuitively, when the agent is very patient ($r\rightarrow 0$) or when news arrives rapidly ($\lambda \rightarrow \infty $), regardless of $\alpha$, the agent can accumulate information with no substantial payoff consequences. Even in the classical environment, the agent may suffer payoff losses because she exploits the risky project for a long time, but the payoff consequences are minimal when the agent is very patient. The benefit of disentanglement is therefore small. When the agent is very impatient ($r\rightarrow \infty $) or when news arrives slowly ($\lambda \rightarrow 0$), short-run, or myopic payoffs approximate the agent's payoffs regardless of the level of disentanglement, which hence has little impact. It follows that the payoff consequences of disentanglement can be meaningful only for intermediate values of $r/\lambda $.

As Figure \ref{Figure1} illustrates, the effects of $p_{H}$ are also non-monotonic. Consider first good news settings (depicted in the left panel). Suppose $p_{H}\leq \bar{p}\left( 1\right) $, so that the probability that project $H$ is good is lower than the cutoff in the classic environment. Regardless of the disentanglement level $\alpha$, project $L$ is exploited. The value of disentanglement is then only due to the ability to continue collecting information; it is increasing in the prior $p_{H}$ that project $H$ is good. When $p_{H}>\bar{p}\left( 0\right) =R_{L}/R_{H}$, regardless of the disentanglement level $\alpha$, project $H$ is explored and exploited. Disentanglement is then beneficial only due to the continuation value in the eventuality that no news arrives and the posterior falls below $R_{L}/R_{H}$ when a sufficiently long period transpires without news. The probability of no news is decreasing in $p_{H}$. Thus, the value of disentanglement is decreasing in the region $(\bar{p}\left( 0\right) ,1)$. Consequently, the ability to disentangle exploration from exploitation is most valuable in the $(\bar{p}\left( 1\right) ,\bar{p}\left( 0\right) )$ region. In this region, when $\alpha =1$, the agent exploits a sub-optimal project for its exploration value. Disentanglement limits the payoff loss associated with such exploration.

Consider now bad news settings (depicted in the right panel). As in good news settings, when $p_{H}<\bar{p}\left( 1\right) $, regardless of $\alpha $, project $L$ is exploited. The value of disentanglement is due to the information it affords. This value is increasing in the prior likelihood that project $H$ is good. When $p_{H}>\bar{p}\left( 1\right) $, in the classical environment, with $\alpha =1$, the agent explores and exploits project $H$. Absent news, the agent becomes increasingly optimistic and continues exploring project $H$. This persistence in the explored and exploited project generates a kink in payoffs, noted by KR, which yields the kink seen in Figure \ref{Figure1}. Disentanglement leads the agent to exploit project $L$ for posteriors higher than $\bar{p}\left( 1\right) $, when its expected payoffs are higher than those from project $H$. The benefit from doing so decreases with the probability that project $H$ is, in fact, the better project. When $p_{H}>\bar{p}\left( 0\right) $, regardless of the level of disentanglement, the agent explores and exploits project $H$ and switches the project she exploits only upon seeing bad news. Thus, expected payoffs are independent of $\alpha$ in the region $(\bar{p}\left( 0\right),1)$.

The following corollary summarizes our discussion.

\bigskip

\begin{corollary}[One Safe Project: Comparative Statics]\label{corr1} 
The disentanglement value $\Delta \Pi (p_{H},r/\lambda )$ is non-monotonic in each of its arguments. It is maximized at $p_{H}^{\ast }$ such that $\bar{p}\left( 1\right) <p_{H}^{\ast }<\bar{p}\left(0\right)$ in good news settings and at $p_{H}=\bar{p}\left( 1\right)$ in bad news settings. 
\end{corollary}

\bigskip

In terms of the degree of disentanglement $\alpha$, increasing it tightens the agent's constraint, and this reduces her expected payoffs. However, the relationship between expected payoffs and $\alpha$ is neither concave nor convex. To see this, consider, for instance, the balanced news setting. For any $p_{H}\in(\bar{p}\left( 1\right) ,\frac{R_{L}}{R_{H}})$, there exists $\alpha ^{\ast} $ such that $\bar{p}\left( \alpha ^{\ast }\right) =p_{H}$. Using the monotonicity of $\bar{p}\left( \cdot \right)$ in Proposition \ref{prop1}, at the outset, the agent exploits the risky project $H$ for any $\alpha >\alpha^{\ast }$. Furthermore, in a balanced news setting, the only way the agent updates her posterior, and changes her exploited project, is by receiving news. Therefore, the agent's expected payoffs are constant in $\alpha$ for $\alpha >\alpha ^{\ast }$. However, for $\alpha <\alpha ^{\ast }$, expected payoffs are strictly decreasing and concave in $\alpha $; see the Appendix for details.\footnote{When $\lambda=\lambda_{H}^g = \lambda_{H}^b$ and $\alpha <\alpha ^{\ast }$, the resulting expected payoffs are given by:
\begin{equation*}
R_{L}+\frac{\lambda(1-\alpha)}{r+\lambda(1-\alpha)}p_{H}(R_{H}-R_{L}),
\end{equation*}
which is concave in $\alpha$.} In particular, expected payoffs are neither concave nor convex in $\alpha$ over the interval $[0,1]$.

\section{Two Risky Projects \label{PoissonBandits}}

We now analyze the general case of two risky projects, where $p_{L},p_{H}\in (0,1)$. For tractability, we assume full disentanglement, $\alpha =0$. In this case, the agent's optimal exploitation choices are simple: she always chooses the myopically optimal project, which we term the \textit{favorable} project. That is, project $x$ is favorable, while project $y$ is unfavorable, if $p_x R_x > p_y R_y$. Both projects are favorable when their expected values coincide. The focus of our analysis is, therefore, on the characterization of optimal exploration. We show that the optimal policy entails very few switches of either the exploited or the explored project. However, unlike the special case in which one project is safe, the information structure has a substantial impact on the characterization of the optimal policy. Furthermore, the optimal policy cannot be characterized via an index \`{a} la \cite{gittins1979bandit}.

We divide our analysis into three subcases. We first discuss balanced news settings. We then consider good news settings. We conclude with our analysis of bad news settings. In all these settings, if the agent receives news that project $H$ is good, then there is no additional value of exploration. If the agent receives news that project $L$ is good, the optimal policy proceeds as described in Section \ref{SafeProject}. Therefore, in what follows, we emphasize features of the optimal policy before news arrives.

\subsection{Balanced News Settings\label{BalancedNews}}

We start by analyzing balanced news settings in which $\lambda_{z}^{b}=\lambda_{z}^{g}=\lambda_{z}$ for $z=L,H$. The analysis of these settings proves instrumental for the characterization of optimal policies in good and bad news settings, which follow. Substantively, while rarely studied in the literature, these settings reflect environments in which the arrival rate of news does not depend on its valence. For example, when assessing the efficacy of a menu of medical treatments using clinical trials, the arrival rate of news depends on the number of patients and the rate at which they are treated, but not necessarily on the quality of the treatments per se. Similarly, when researching the promise of an investment opportunity, the arrival rate of news often depends on the scope and speed of investigation, not explicitly on the quality of the investment option.

Suppose the agent optimally explores one of the projects at the outset. Absent news, the agent's posterior probabilities and, therefore, her decision problem do not change. In particular, in the optimal policy, the agent does not switch the project she explores unless news arrives. The agent's exploration choice is then effectively a static problem corresponding to her decision of which project to start exploring at the outset.

In order to characterize the optimal policy, it is useful to consider a modification of the probability that any project $x=L,H$ is good, which we denote by $\tilde{p}_{x} \geq p_{x}$. When project $x$ is favorable, we define $\tilde{p}_{x}\equiv p_{x}$. When project $x$ is unfavorable, we define $\tilde{p}_{x}\equiv \min (p_{y}R_{y}/R_{x},1)$.\footnote{In this case, $p_{y}R_{y}\geq p_{x}R_{x}$, and thus $\tilde{p} _{x}\geq p_{x}$.} When the agent is indifferent between exploiting either project myopically, so that both projects are favorable, the two definitions coincide.

\bigskip

\begin{proposition}[Optimal Exploration in Balanced News Settings]\label{prop2}
Suppose $\lambda_{z}^{b}=\lambda_{z}^{g}=\lambda_{z}$ for $ z=L,H$. Any optimal exploration strategy entails exploring project  $x$ until good news arrives, where $\lambda_{x}(1-\tilde{p} _{x})\geq \lambda_{y}(1-\tilde{p}_{y})$, with $y\neq x$.  
\end{proposition}

\bigskip

Intuitively, the agent selects the project that is most ``informative.'' A higher arrival rate of news certainly increases the appeal of exploring a project. In addition, information is useful only when it affects exploitation decisions. When the agent explores the favorable project, only bad news triggers a switch in exploitation. Bad news on project $x$ can arrive only for a bad project $x$, which occurs with probability $1-p_{x}$. In contrast, exploration of an unfavorable project $y$ may or may not lead to a change in exploitation choices, even if good news arrives. Indeed, if the agent is sufficiently optimistic about project $x$, good news on project $y$ would not sway her exploitation choices. In such cases, exploring project $y$ 
before learning the quality of project $x$ is of no value. 
Hence, the probability adjustment factor in the proposition, which raises the hurdle for unfavorable projects.

The optimal exploration strategy is generally unique, with two exceptions. First, whenever the knife-edge condition that $\lambda_{x}(1-\tilde{p}_{x})=\lambda_{y}(1-\tilde{p}_{y})$ for $y\neq x$ holds, any exploration strategy is optimal. Second, if project $H$ is explored and good news arrives, the agent exploits project $H$ forever. Any ensuing exploration is then optimal.

In the classical environment, when exploration and exploitation are entangled, each project is associated with a (Gittins) index that depends only on the parameters of that project. Specifically, the index for a project $z$ is given by $p_{z}R_{z}\frac{(r+\lambda_{z})}{(r+p_{z}\lambda_{z})}$. The agent explores and exploits the project with the higher index. When news is balanced, the agent switches away from exploring and exploiting project $z$ only upon receiving news.

In our environment, with exploration disentangled from exploitation, the expected reward $p_{x}R_{x}$ of each project $x$ serves as a separable index for exploitation: the agent optimally exploits whichever project generates the highest expected reward. The agent may, however, switch her exploited project twice when exploration starts from an unfavorable project $L$: first, if she learns her initially unfavorable project $L$ is good and, second, if she later learns her initially exploited project $H$ is, in fact, good (as $R_{H}>R_{L}$). This already highlights the importance of disentanglement, as exploration and exploitation need not track one another. Furthermore, as Proposition \ref{prop2} suggests, there is no obvious separable index that underlies optimal exploration, a point we return to in the next subsection. Intuitively, the value of exploring the unfavorable project depends on the returns of the favorable project.

Comparative statics are clearly affected by the ability to disentangle exploration from exploitation. Under the canonical assumption that the two are entangled, a higher prior probability that one project is good makes it more appealing for exploration and exploitation. In contrast, as Proposition \ref{prop2} indicates, in our setting, a higher prior that a project is good may make its exploration \textit{less} appealing. Additionally, optimal exploration depends only on the ``informational value'' derived from each project. Consequently, unlike in the classical environment, the optimal policy does not depend on the discount factor in ours.

\subsection{No Exploration Index}

As mentioned above, in the classical environment, Gittins (1979)'s characterization of the optimal policy holds. That is, each project is associated with an index that only depends on the parameters of that project. At any point, the agent explores and exploits the project with the highest current index. While Proposition \ref{prop2} offers a simple characterization of the optimal policy, we now show that, in our setting, optimal exploration is not governed by an index \`{a} la Gittins (1979).

Suppose that the optimal policy in a balanced news setting can be described via an index tailored to each project. We denote by $I(p,R,\lambda)$ the index corresponding to a project with a probability $p$ of being good, an arbitrary reward $R>0$ conditional on being good, and a rate of news arrival---good or bad---of $\lambda$.

Consider three hypothetical projects. Project $i=1,2,3$ is governed by a probability $p_{i}$ that it is good, associated with a flow reward of $R_{i}>0$, and a news arrival rate of $\lambda_{i}>0$. Suppose that 
\begin{equation*} 
p_{2}R_{2}>p_{1}R_{1}\text{ \ \ and \ \ }\lambda_{2}(1-p_{2})<\lambda_{1}\left( 1-\frac{p_{2}R_{2}}{R_{1}}\right). 
\end{equation*} 
Then, using Proposition \ref{prop2}, when the agent has access to projects $1$ and $2$ , she optimally exploits project $2$, but explores project $1$. That is, $ I(p_{1},R_{1},\lambda_{1})>I(p_{2},R_{2},\lambda_{2})$.

Suppose now that
\begin{equation*} 
p_{2}R_{2}>R_{3}>p_{3}R_{3}>p_{1}R_{1}. 
\end{equation*} 
This implies that, when the agent has access to projects $2$ and $3$, she optimally explores and exploits project $2$. That is, $I(p_{2},R_{2},\lambda_{2})>I(p_{3},R_{3},\lambda_{3})$.

Suppose, further, that $\lambda_{3}$ is high enough so that
\begin{equation*} 
\lambda_{3}(1-p_{3})>\lambda_{1}\left( 1-\frac{p_{3}R_{3}}{R_{1}}\right). 
\end{equation*} 
This implies that, when the agent has access to projects $1$ and $3$, she explores and exploits project $3$. Therefore, $I(p_{3},R_{3},\lambda_{3})>I(p_{1},R_{1},\lambda_{1})$, establishing a cycle, in contradiction. Although this construction is done for the balanced news setting, it is robust to small perturbations of parameters. In particular, the optimal exploration policy is not generally governed by an index for either good or bad news settings either. Thus,

\bigskip

\begin{corollary}[No Exploration Index]\label{corr2}
The optimal exploration policy is not governed by an index.   
\end{corollary}

\bigskip

We stress that this conclusion is not driven by an excess number of degrees of freedom. The classical environment entails the same project characteristics and, therefore, the same degrees of freedom.

\subsection{Good News Settings}

We now analyze good news settings. Before describing our general characterization, consider the following example, highlighting the impacts of disentanglement when both projects are risky.

\begin{description} \item[Example 1 (Good News: Ex-ante Identical Projects)] Suppose the two projects are ex-ante identical: $p_{L}=p_{H}$ and $R_{L}=R_{H}$.\footnote{Strictly speaking, this violates our assumption that $R_{H}>R_{L}$, which generates the non-trivial scenarios for Section \ref{SafeProject}. We assume equal rewards here to simplify our illustration of the stark effects of disentanglement.} Furthermore, for simplicity, consider the pure good news setting in which $ \lambda_{z}^{b}=0$ and $\lambda_{z}^{g}=\lambda>0$ for $z=L,H$.

In the classical environment with $\alpha =1$, the optimal strategy requires splitting exploration and exploitation equally between the two projects until receiving news. Intuitively, consider a discrete time approximation of this problem. If the agent explores and exploits project $x$, the corresponding Gittins index declines absent news---the agent becomes more pessimistic about project $x$. She should then immediately switch to project $y$. In the limit, splitting exploration and exploitation equally across the two projects leads the two indices to decline at the same rate and maintains the incentive to continue with such a split. We can interpret this strategy as requiring the agent to switch between projects infinitely often.\footnote{See case (v) in Section 3.3.2 of \cite{gittins2011bandits} for details.}

In contrast, in our setting with $\alpha =0$, an optimal policy requires indefinite disentanglement, i.e, exploiting one project and exploring the other indefinitely, until the arrival of good news. If the agent exploits project $x$ and explores project $y$ at the outset for any infinitesimal time interval, project $x$ becomes favorable, so continuing to exploit project $x$ is optimal. Furthermore, information is useful to the agent only if it leads her to change her exploited project. Good news on project $x$ would not alter her exploitation choices; only good news on the unfavorable project would. This means that it is optimal for the agent to use her entire exploration budget on project $y$: any splitting of exploration resources between the two projects is sub-optimal since it reduces the effective rate at which good news arrives on the unfavorable project. As a consequence, with full disentanglement, the agent switches her exploitation choices at most once and \textit{never} switches her exploration choice prior to receiving news. Of course, the agent is indifferent as to which project she explores and which she exploits at the outset given the complete symmetry of the problem. In fact, the agent can also choose at random which project to start exploring. The contrast with the classical environment is that such randomization cannot proceed with a split of exploration or exploitation for a non-trivial duration. 
\end{description}

In general, in the classic environment, when projects are heterogeneous, the agent initially explores and exploits the project with the higher Gittins index. Absent news, that project's Gittins index declines over time, until it reaches equality with the index of the other project. Upon such indifference, the agent splits exploration and exploitation to maintain her indifference. We can interpret this splitting of attention, or exploration resources, as the limit of sequential immediate switches in discrete time \citep[see][]{gittins2011bandits}. As we now show, such rapid switches \textit{never} occur when exploration and exploitation are disentangled.

Consider then a disentangled setting with good (or balanced) news, where $\lambda_{x}^{g}\geq \lambda_{x}^{b}$ for $x=L,H$. Whenever project $x$ is favorable, so that $p_{x}R_{x}\geq p_{y}R_{y}$, exploiting project $x$ is optimal. When the agent explores project $x$, receiving no news makes her increasingly pessimistic. We denote by $\bar{t}_{x}(p_{L},p_{H})$ the time it takes the agent to reach indifference between the expected values of both projects. If $p_{x}R_{x}=p_{y}R_{y}$, then $\bar{t}_{x}(p_{L},p_{H})=0$; otherwise, $\bar{t}_{x}(p_{L},p_{H})>0$.\footnote{If $p_{x}R_{x}>p_{y}R_{y}$ in a balanced news setting, exploring project $x$ does not change the agent's posterior and we set $\bar{t}_{x}(p_{L},p_{H})=\infty$. When $p_{x}R_{x} \leq p_{y}R_{y}$, we denote $\bar{t}_{x}(p_{L},p_{H})=0$ even when $\lambda_{x}^{g}=\lambda_{x}^{b}$ and the agent does not alter her prior as time passes without information.} 
Specifically, after exploring project $x$ for a duration $\bar{t} _{x}(p_{L},p_{H})$ without receiving news, the agent's posterior that project $x$ is good is precisely $p_yR_y/R_x$. That is,
\begin{equation*} 
\frac{p_{x}e^{-\lambda_{x}^{g}\bar{t}_{x}(p_{L},p_{H})}}{p_{x}e^{-\lambda_{x}^{g}\bar{t}_{x}(p_{L},p_{H})}+(1-p_{x})e^{-\lambda_{x}^{b}\bar{t} _{x}(p_{L},p_{H})}}=p_yR_y/R_x. 
\end{equation*} 
Simplifying, whenever $\lambda_{z}^{g}>\lambda_{z}^{b}$, we obtain: 
\begin{equation*}
\bar{t}_{x}(p_{L},p_{H})=\frac{1}{\lambda_{x}^{g}-\lambda_{x}^{b}}\ln \left( \frac{p_{x}(R_{x}-p_{y}R_{y})}{p_{y}R_{y}(1-p_{x})}\right).
\end{equation*} 

\bigskip

We now state our result characterizing optimal exploration in this setting.\footnote{As stated at the outset, we ignore 0-measure sets. When we say the agent explores project $y$ at some time, we mean the agent explores project $y$ for a positive-measure set of times. Switching to a project $x$ implies that there is an ensuing positive-measure set of times at which the agent explores project $x$.}

\bigskip

\begin{proposition}[Optimal Exploration in Good News Settings]\label{prop3}
Suppose $\lambda_{z}^{g}>\lambda_{z}^{b}$ for $z=L,H$ and that project $x$ is favorable, so that $p_{x}R_{x}\geq p_{y}R_{y}$. An optimal exploration strategy is described as follows.
\begin{itemize} 
\item If, at any time the agent explores project $y$, she never switches absent news. 
\item If the agent initially explores project $x$, then if, absent news, she switches to exploring project $y$, she does so at a time $T\leq \bar{t}_{x}(p_{L},p_{H})$. 
\end{itemize}
Furthermore, if $\lambda_{x}^{b}=0$, there is an optimal strategy in which the agent never switches her explored project absent news.
\end{proposition}

\bigskip

Proposition \ref{prop3} illustrates that disentanglement dramatically reduces the expected number of switches prescribed by the optimal policy. The exploited project can be switched at most twice: starting from project $H$, good news about project $L$ could lead to one switch if the agent is sufficiently pessimistic about project $H$, and later good news about project $H$ could lead to a switch back to project $H$. The explored project can be switched at most once before any news arrives. In fact, if the agent explores the unfavorable project initially, she never switches the explored project absent news, no matter how pessimistic she becomes about this project. Of course, when she becomes pessimistic about the explored project, she also becomes increasingly confident that she is exploiting the superior project. 

The role of disentanglement is evident in the optimal policy described in Proposition \ref{prop3}. Under this policy, eventually, absent news, the exploited project must differ from the explored project. Indeed, since the optimal policy prescribes indefinite exploration of a project, eventually the posterior probability that the explored project is good must be low enough to make the other project favorable and, therefore, exploited.

\begin{figure}[t] 
\includegraphics[width=1\columnwidth]{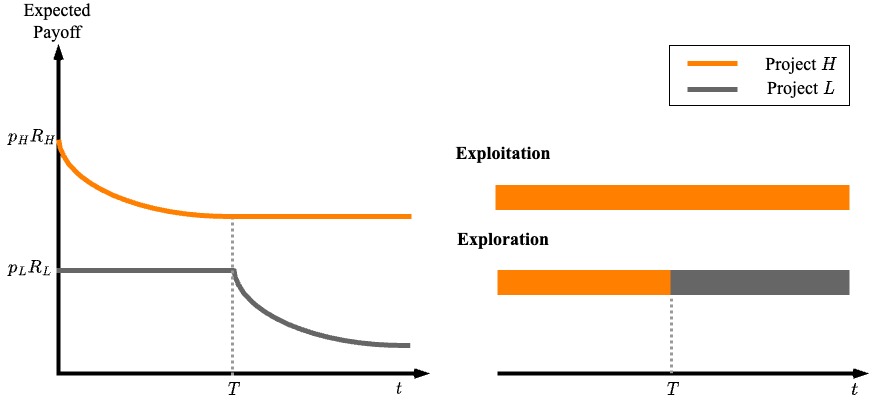} \vspace{1mm} 
\caption{Optimal policy with two risky projects in good news settings\label{FigureProp3}} 
\end{figure}

Figure \ref{FigureProp3} depicts the exploration and exploitation patterns in good news settings. In the figure, project $H$ is initially favored and both explored and exploited. As time progresses without news, the agent's confidence in project $H$ diminishes. However, at time $T$, the agent switches to exploring project $L$ even though project $H$ remains the more favorable option. In the absence of news, the agent continues to exploit project $H$ indefinitely. Starting from time $T$, the projects she explores and exploits diverge.

In order to understand the logic of Proposition \ref{prop3}, consider first the case in which the optimal exploration strategy prescribes exploring the unfavorable project $y$ initially, so that the agent explores and exploits different projects. The value of exploring project $y$ depends only on the rate at which good news arrives: receiving bad news on project $y$ retains project $x$ as favorable. Thus, similar to the setting with one safe project, the value of exploring project $y$ depends on $\lambda_{y}^{g}$, but not on $\lambda_{y}^{b}$. In particular, this value is the same if we increase the arrival rate of bad news so that news is balanced, $\lambda_{y}^{b}=\lambda_{y}^{g}$. In this case, as discussed in Section \ref{BalancedNews}, the agent's posteriors do not change absent news. If it is optimal to explore project $y$ at some point, it is also optimal to explore project $y$ after any amount of time that has passed without news. We can conclude that it must also be optimal to continue exploring project $y$ in the absence of news when $\lambda_{y}^{b}<\lambda_{y}^{g}$.

Consider now the case in which the optimal strategy prescribes exploring the favorable project $x$ initially, so that the agent explores and exploits the same project at first. Why can it be optimal for the agent to switch the project she explores when $\lambda_{x}^{b}>0$? Suppose, for instance, that $p_{H}R_{H}>R_{L}\geq p_{L}R_{L}$. In this scenario, exploring project $L$ initially is not useful: even good news on project $L$ would not lead the agent to switch her exploited project. Instead, if $\lambda_{H}^{b}>0$ and the agent explores project $H$, she would switch to exploiting project $L$ upon receiving bad news on project $H$: exploring project $H$ is valuable. When, instead, $p_{L}R_{L}<p_{H}R_{H}<R_{L}$, good news on project $L$ would lead the agent to switch to exploiting project $L$, implying that exploring project $L$ can be useful. The determination of the switching time $T$ depends on the relative magnitudes of $p_{H}R_{H}$, $p_{L}R_{L}$, and the arrival rates of news on the two projects.

Why can the agent not switch exploration of the favorable project $x$ after a duration $T>\bar{t}_{x}(p_{L},p_{H})$ without receiving news? By the definition of $\bar{t}_{x}(p_{L},p_{H})$, after such a duration $T$ without news, project $x$ becomes unfavorable. Our previous arguments then imply that, absent news, indefinite exploration of project $x$ beyond time $T$ is optimal.\bigskip

As is common in individual decision problems, multiplicity of optimal policies is rare, as the following claim illustrates. 

\bigskip

\begin{claim}[Uniqueness of Optimal Policy]\label{claim1}
Suppose that $\lambda_{z}^{g}>\lambda_{z}^{b}$ for $z=L,H$ and that the favorable project $x$ is such that $R_y>p_x R_x > p_y R_y$. Then, generically, the optimal policy is unique.\footnote{The optimal policy is unique as long as $\lambda_x^b R_y \neq \lambda_y^g R_x$, which excludes a zero-measure set of parameters.}   
\end{claim}

\bigskip

The claim captures settings in which good news on the unfavorable project is immediately valuable: it alters the exploited project. Multiplicity can arise when information has no instantaneous value, regardless of which project is explored. This is the case when $p_{H}R_{H}>R_{L}$ in a pure good news setting, with $\lambda^b_z=0$ for $z=L,H$. In this parameter region, the agent exploits project $H$ initially even after good news on project $L$. Therefore, exploration has no instantaneous value. Of course, good news on project $L$ leads the agent to explore project $H$, and a sufficiently long time without news would lead the agent to switch to exploiting project $L$. In this case, there are multiple optimal strategies that differ in which project is explored initially.

\bigskip

We now turn to a discussion of the initial exploration choice. For expositional simplicity, we focus on the special case of pure good news settings, where $\lambda_{x}^{b}=0$, $x=L,H$. In this case, Proposition \ref{prop3} indicates that an optimal policy has the agent explore the same project until receiving news, implying that the initial choice is permanent absent news. We also restrict attention to the case $p_{L}R_{L}<p_{H}R_{H}<R_{L}$, where information on both projects is valuable at the outset. Indeed, exploring project $H$ for a sufficiently long time would make the agent pessimistic about the quality of that project and, absent news, the agent would switch her exploited project after a duration $\bar{t}_{H}(p_{L},p_{H})$. Exploring project $L$ is also valuable: receiving good news on that project would lead the agent to immediately switch the project she exploits. In particular, for this set of parameters, exploring either project can be optimal depending on the difference between news' arrival rates. 

In line with our previous notation, we denote $\tilde{p}_{L}\equiv p_{H}R_{H}/R_{L}$. Thus, $\tilde{p}_{L}$ corresponds to the prior that project $L$ is good at which the agent is indifferent between the two projects.

\bigskip

\begin{claim}[Initial Choice with Pure Good News]\label{claim2}
Suppose $\lambda_{z}^{b}=0$ for $z=L,H$ and that $p_{L}R_{L}<p_{H}R_{H}<R_{L}$. It is optimal to explore project $H$ if and only if $\lambda_{H}^{g}\frac{w-\rho _{L}}{1-\rho _{L}} (1-p_{H})\geq \lambda_{L}^{g}(1-\tilde{p}_{L})$, where $w=e^{-r\bar{t}_{H}(p_{L},p_{H})}$ and $\rho _{L}=\lambda_{L}^{g}/\left( r+\lambda_{L}^{g}\right)$.   
\end{claim}

\bigskip

The specification in the claim is reminiscent of the one appearing in Proposition \ref{prop2}, with the added multiplier $\frac{w-\rho _{L}}{1-\rho _{L}}$ for project $H$. As already noted, Proposition \ref{prop3} indicates that in the pure good news setting, we only need to compare two cases, differing in which project is explored until news. Suppose that exploring project $H$ is optimal. At time $\bar{t}_{H}(p_{L},p_{H})$, project $L$ becomes favorable and the agent explores and exploits different projects. As described in the intuition for Proposition \ref{prop3}, the value of exploring project $H$ depends only on the arrival rate of good news: receiving bad news on project $H$ sustains project $L$ as favorable. Thus, the value of exploring project $H$ depends on $\lambda_{H}^{g}$, but not on $\lambda_{H}^{b}$. Consequently, starting at $\bar{t}_{H}(p_{L},p_{H})$, the expected payoffs from this problem are the same as those in an auxiliary balanced news problem with arrival rate of $\lambda_{H}^{g}$ for both good and bad news. The determination of which project to explore must then conform with the characterization in Proposition \ref{prop2}. 

In contrast with the balanced news setting, when $p_{H}R_{H}>p_{L}R_{L}$, the initial comparison includes the factor $\frac{w-\rho _{L}}{1-\rho _{L}}$, penalizing the exploration of project $H$. To understand this penalty, note that, absent news, if the agent explores project $H$, she switches the exploited project only after a duration $\bar{t}_{H}(p_{L},p_{H})$. The larger this duration, the longer the period in which exploration without news does not affect the agent's exploitation, and the less appealing it is to explore project $H$. Indeed, $w$ and $\frac{w-\rho _{L}}{1-\rho _{L}}$ decrease with $\bar{t}_{H}(p_{L},p_{H})$. If both projects are favorable, so that $\bar{t}_{H}(p_{L},p_{H})=0$, or if the agent is infinitely patient ($r=0$), then $w=1$ and the claim's inequality boils down to the comparison in Proposition \ref{prop2}. Similar characterizations hold for other cases of prior probabilities that either project is good.

This claim offers another way to show the way by which disentaglement of exploration from exploitation has bite. Although project $H$ is optimally exploited at the outset, it is optimal to explore project $L$ whenever $\rho _{L}>w$, i.e., when news arrival on project $L$ is fairly rapid. Similarly, as the agent becomes more and more impatient, with $r$ increasing indefinitely, both $w$ and $\rho _{L}$ approach $0$, and the agent explores project $L$. Since $p_{L}R_{L}<p_{H}R_{H}<R_{L}$, in these circumstances, the agent would exploit project $H$ initially regardless of which project she explores. She switches the project she exploits only if she learns that project $L$ is good. Furthermore, unlike the comparative statics in the classical entangled environment, exploration of project $L$ becomes more appealing as $p_{H}$ increases.

In general, comparing the payoffs generated by the optimal policy in our setting to those generated in the classical environment yields similar insights to those observed when one of the projects is safe, as presented in Corollary \ref{corr1}. When arrival rates $\lambda_{L}^{g}$ and $\lambda_{H}^{g}$ are very high or when the discount rates are very low, the agent can achieve payoffs close to those corresponding to a complete information setting in both environments. Similarly, when arrival rates $\lambda_{L}^{g}$ and $ \lambda_{H}^{g}$ are very low, or discount rates are very high, the agent receives an expected payoff approximating the myopic expected payoff in both environments. In particular, the benefits of disentanglement are most pronounced for intermediate levels of arrival and discount rates. Similarly, the benefits of disentanglement are non-monotonic in the prior $p_{H}$.

\subsection{Bad News Settings}

We now turn to bad news settings. Before characterizing the optimal policy, consider the following example, which complements Example 1 and illustrates some of the qualitative differences between the information structures we consider.

\begin{description} \item[Example 2 (Bad News: Project $L$ is Favorable)] Suppose that $\lambda_{z}^{g}=0$ and that $\lambda_{z}^{b}=\lambda_{z}>0$ for $z=L,H$. Furthermore, suppose project $L$ is favorable, so that $p_{L}R_{L}>p_{H}R_{H}$. 

In the classical bandit environment, if the wedge between the projects' expected values is sufficiently high, the agent explores and exploits project $L$. Absent news, the agent becomes increasingly optimistic about project $L$ and thus continues exploring and exploiting project $L$ indefinitely. If project $L$ is indeed good, then the agent never receives bad news on project $L$ and therefore never learns whether project $H$ is good.

In contrast, with full disentanglement, even if the agent explores project $L$ at the outset, which is optimal if $\lambda_{L}$ is high enough, she does not do so indefinitely. Switching the exploited project can occur both upon learning that project $L$ is bad, and when becoming increasingly optimistic about project $H$. As the duration of exploration of project $L$ increases, so does the posterior $p_{L}$, implying that the likelihood of learning that project $L$ is bad vanishes, as does the value of exploring it. Consequently, switching to exploring project $H$ is eventually optimal. Thus, disentanglement is not only useful but, in bad news settings, may lead to more switching of the explored projects than in the classical environment. \end{description}

The observation in Example 2 that, in the classical environment, the agent explores and exploits the same project indefinitely unless news arrives is clearly quite general. At the outset, if the Gittins index is higher for project $x$, that project is explored and exploited. Absent bad news, the agent becomes more optimistic about the quality of project $x$ and its associated Gittins index increases. In contrast, in our environment, when exploration and exploitation are disentangled, the agent may optimally switch the projects she explores.

\bigskip

\begin{proposition}[Optimal Exploration in Bad News Settings]\label{prop4}
Suppose $\lambda_{z}^{b}>\lambda_{z}^{g}$ for $z=L,H$. The optimal exploration strategy is described as follows. 
\begin{itemize} 
\item If the agent initially explores project $H$, she never switches absent news. 
\item If the agent initially explores project $L$, she switches after a period $T<\infty$ without news. 
\end{itemize}
\end{proposition}

\bigskip

In contrast with the optimal policy in good news settings, characterized in Proposition \ref{prop3}, in bad news settings, the optimal policy never entails exploring project $L$ forever. The intuition is similar to that appearing in Example 2. When the agent explores project $L$, absent news, she becomes increasingly optimistic about its prospect. Consequently, regardless of news' arrival rates, after a sufficiently long period of exploring project $L$ without news, the agent exploits project $L$ and the likelihood she learns project $L$ is bad becomes vanishingly small. The value of exploring project $H$, however, remains strictly positive.

Proposition \ref{prop4} also highlights the fact that in bad news settings---unlike in good news settings---disentanglement only plays a role in the relative short run. After enough time passes without news, the agent necessarily explores and exploits project $H$.

\begin{figure}[t] 
\includegraphics[width=1\columnwidth]{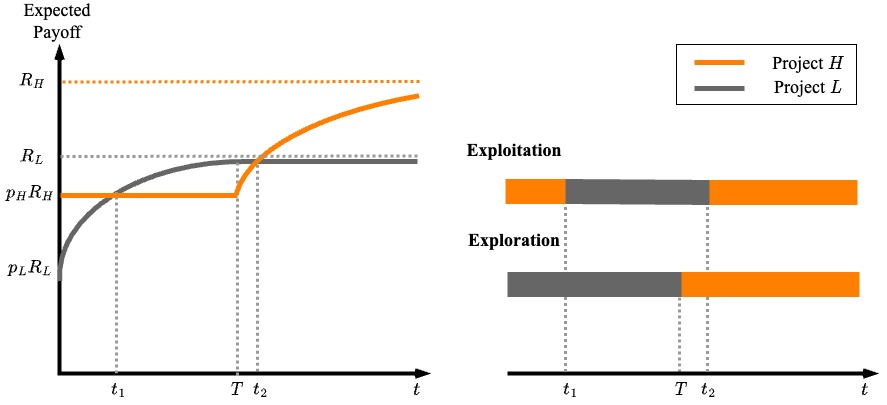} \vspace{1mm} 
\caption{Optimal policy with two risky projects in bad news settings\label{FigureProp4}} 
\end{figure}

Figure \ref{FigureProp4} depicts the exploration and exploitation patterns in bad news settings. In the figure, project $H$ is initially favored and therefore exploited, but project $L$ is explored, say, because it features high news arrival rate. As time progresses without news, the agent's confidence in project $L$ increases. At time $t_1$, project $L$ becomes favored, and the agent switches to exploiting it. By time $T$, the rate of learning on project $L$ has flattened, and the agent switches to exploring project $H$, while continuing to exploit project $L$. Without news, the agent becomes increasingly optimistic about project $H$. At time $t_2$, project $H$ becomes favorable again and the agent switches to exploiting it. Absent news, the agent continues to exploit and explore project $H$ indefinitely.

In general, the proof that the optimal policy entails no switching when project $H$ is explored initially is involved. When $p_{H}R_{H}\geq R_{L}$, the claim follows immediately. In this case, the agent must explore project $H$ from the start since even good news about project $L$ would not lead her to change the project she exploits; exploring project $L$ is not decision relevant. 

When $R_{L}>p_{H}R_{H}\geq p_{L}R_{L}$, it is useful to consider an auxiliary problem in which the agent receives balanced news about project $H$ at the original rate $\lambda_{H}^{b}$ for both good and bad news; the original arrival rates are used when project $L$ is explored. In the auxiliary problem, the agent has more information than in the original problem. If, in the original problem, exploring project $H$ is optimal, then it is also optimal to explore project $H$ in the auxiliary problem, where it is more informative. Absent news, the agent would optimally explore project $H$ until news in the auxiliary problem: her posteriors do not change. The agent can emulate that same strategy even in the original problem. Furthermore, exploring and exploiting project $H$ until news generates the same payoffs in both problems. Since it is optimal in the auxiliary problem, which affords the agent more information, it must be optimal in the original problem as well. The remaining case in which $p_{L}R_{L}>p_{H}R_{H}$ is discussed in the Appendix. The proof of Proposition \ref{prop4} also illustrates that the optimal policy is unique until news arrives.\footnote{If any news arrives about project $H$, or if bad news arrives about project $L$, later exploration has no effect on payoffs. In particular, after such news arrives, any exploration strategy is optimal.}

As for the initial choice of projects, our discussion above suggests that whenever $p_{H}R_{H}\geq R_{L}$, the agent begins by exploring project $H$. When $R_{L}>p_{H}R_{H}$, project $L$ is explored initially when bad news' arrival rate for project $L$ is sufficiently high. The proof of Proposition \ref{prop4} provides the relevant parameter comparisons governing the choice of which project is optimally explored first.

\bigskip

In terms of comparative statics, in bad news settings---unlike good news settings---as $r$ grows indefinitely, the agent optimally explores and exploits the same project: the only news that would change short-term exploitation is bad news on the exploited project. The comparison of payoffs generated by the optimal policy with and without disentanglement is similar to that observed for good news settings and that described with one safe project in Corollary \ref{corr1}. In particular, the benefits of disentanglement are most pronounced for intermediate values of parameters.

\section{Concluding Remarks}
 
This paper presents a new framework for studying experimentation that, unlike the conventional multi-arm bandit paradigm, permits agents to disentangle exploration from exploitation. Our findings are applicable to the extensively studied case of Poisson bandits, accommodating multiple risky projects and general good and bad news settings. We demonstrate that the optimal policy entails full learning asymptotically, displays significant persistence, yet cannot be discerned through an index like Gittins'. The ability to disentangle proves especially beneficial for intermediate parameter values.

Absent news, disentanglement is utilized at different phases of the experimentation process, depending on the format of news. In good news settings, the agent optimally explores and exploits different projects after enough time has passed. In bad news settings, the agent may explore and exploit different projects only in the short run. After a long duration, she necessarily explores and exploits the same project.

We hope our framework can be used for a variety of applications that have been investigated only through the lens of the classical bandit environment, including team experimentation \citep[as in][]{keller2005strategic, keller2010strategic, strulovici2010learning}, expert delegation \citep[as in][]{guo2016dynamic}, job search \citep[as in][]{jovanovic1979job, miller1984job}, and more. 

\section{Appendix}

\subsection{Preliminaries}

\begin{proof}[Proof of Proposition \ref{prop0}] Let $U_{t}$ be the continuation payoff according to the optimal policy at time $t$, and let $V$ denote the full-information payoff, when the realized quality of each project is known.

Denote by $M_{t}$ be the myopic payoff---the value of the favorable project---given the information the agent has at time $t$ under the optimal policy. 
Let $m_{t}=\mathbf{E}M_{t}$. The value $m_t$ is increasing in $t$ since the agent's information improves over time. Also, $m_t\leq \mathbf{E}V$. Therefore, the limit $m_{\infty }=\lim_{t\rightarrow \infty }m_{t}$ exists.

Let $\varepsilon >0$ and $T$ be sufficiently large so that, by exploring both projects at the same rate for a period of time $T$, the agent can achieve a continuation payoff of $\mathbf{E}V-\varepsilon $.

The following inequalities must hold for all $t$: 
\begin{equation*} 
(1-e^{-rT/(1-\alpha )})m_{t}+e^{-rT/(1-\alpha )}(\mathbf{E}V-\varepsilon)\leq \mathbf{E}U_{t}\leq r\int_{\tau =0}^{\infty }e^{-r\tau }m_{t+\tau}\leq m_{\infty }. 
\end{equation*} 
The left inequality follows from the fact that the agent can use the optimal exploration policy and exploit the favorable project for a period of time $T/(1-\alpha )$ and use her exploration resources over that duration to achieve continuation payoff of at least $\mathbf{E}V-\varepsilon $ from time $t+T/(1-\alpha )$ onwards. The right-most inequality follows from the fact that, with any strategy, the conditional expectation of the flow payoff at time $t+\tau$ is smaller than the conditional expectation of the myopic payoff at that time: the most the agent can get at any time $t+\tau $ is $m_{t+\tau}$.

By taking the limit $t\rightarrow \infty $ we obtain that $m_{\infty }\geq \mathbf{E}V-\varepsilon $. Since this is true for every $\varepsilon $ and since $m_{t}\leq \mathbf{E}V$ for all $t$, we get that $m_{\infty }=\mathbf{E}V$. It follows that $\lim_{t\rightarrow \infty }\mathbf{E}U_{t}=\mathbf{E}V$. Finally, since $U_{t}$ is also a sub-martingale, $U_{t}$ must converge to $V$ almost surely, as desired. \end{proof}

\bigskip

For much of our analysis, it will be useful to note that when an agent discounts at a rate $r$ and receives news arriving at a rate $\lambda $, the expected discount at the time $\tilde{t}_{\lambda }$ at which news first arrives is given by: 
\begin{equation} 
\mathbb{E}(e^{-r\tilde{t}_{\lambda }})=\int_{0}^{\infty }\lambda e^{-\lambda t}e^{-rt}dt=\frac{\lambda }{r+\lambda }. \label{Kroovit} 
\end{equation}

\subsection{One Safe Project: Proofs and Additional Analysis}

\begin{proof}[Proof of Proposition \ref{prop1}] Consider decision problem $\Gamma _{B}$ with balanced news arriving at a rate of $\lambda =\lambda_{H}^{g}=\lambda_{H}^{b}$. Absent news, the posterior that the risky project is good remains constant. Thus, the time elapsed exploring a project does not affect which project should be exploited, implying that the optimal strategy is constant as long as no news arrives. If the agent explores and exploits project $H$ until news, the resulting expected payoff is: 
\begin{equation*} 
pR_{H}+(1-p)\frac{\lambda }{r+\lambda }R_{L}. 
\end{equation*} 
The first term corresponds to a realized good project $H$, where, regardless of news, the agent gets rewards. The second term corresponds to a realized bad project $H$. The agent switches to project $L$ only when bad news about project $H$ arrives, with an expected discount of $\frac{\lambda }{r+\lambda}$. Analogous logic implies that the payoff of exploiting project $L$ while exploring project $H$ at a rate of $1-\alpha $ is: 
\begin{equation*} 
R_{L}+p\frac{\lambda (1-\alpha )}{r+\lambda (1-\alpha )}(R_{H}-R_{L}).
\end{equation*} 
In particular, since the difference in payoffs between exploiting project $H$ and project $L$ is monotonic in $p$, there is a cutoff $\bar{p}\left( \alpha\right)$ such that if $p>$ $\bar{p}\left( \alpha \right) $, the DM explores and exploits project $H$ until news, while if $p<\bar{p}\left( \alpha \right)$, the DM explores and exploits project $L$ indefinitely. At the cutoff $\bar{p}\left( \alpha \right) $, the DM is indifferent. Equating the two expected payoff expressions yields the value of $\bar{p}\left( \alpha \right)$ in the statement of the proposition.

We now move to a good-news setting $\Gamma _{G}$, where $\lambda_{H}^{g}>\lambda_{H}^{b}$, so that $\lambda =\max \{\lambda_{H}^{g},\lambda_{H}^{b}\}=\lambda_{H}^{g}$. We claim that a strategy of the following form must be optimal: there is a $T$ (or a $\hat{p}$) such that, absent news, the agent exploits project $H$ for $t<T$ (or $p>\hat{p}$) and exploits project $L$ for $t\geq T$ ($p\leq \hat{p}$), where $T$ is such that exploration of project $H$ for a duration of $T$ leads $p$ to decline to $\hat{p}$. Consider an auxiliary problem $\Gamma _{A}$ with the following modified news process: if the agent exploits project $L$ and explores project $H$ at a rate of $1-\alpha $, she receives both good and bad news on project $H$ at a rate of $\lambda_{H}^{g}$; If the agent exploits project $H$, she receives news as prescribed in problem $\Gamma _{G}$. The candidate strategy delivers the same payoffs in $\Gamma _{A}$ and in $\Gamma _{G}$: (i) when exploiting project $H$, news arrives at the same rate in both problems; (ii) when exploiting project $L$, the additional arrival rate of bad news is not advantageous since only good news on project $H$ would lead the agent to switch projects. Furthermore, payoffs in $\Gamma_{A}$ must be weakly higher than in $\Gamma _{G}$. Thus, if the candidate strategy is optimal in $\Gamma _{A}$, it must be optimal in $\Gamma_{G}$. We now show that such a strategy is optimal in $\Gamma _{A}$. In this problem, absent news, beliefs do not change when exploiting project $L$. Therefore, if at any point it is optimal to exploit project $L$ at time $t$ in $\Gamma_{A}$, then absent news, it is also optimal to do so at later times.

We want to show that $\hat{p}=\bar{p}\left( \alpha \right)$. The value of exploiting project $L$ is the same in $\Gamma_{G}$ and in $\Gamma_{B}$, with with news arrival rate of $\lambda=\lambda_{H}^{g}$, whereas the value of exploiting project $H$ is lower in $\Gamma _{G}$. Therefore, $\hat{p}\geq \bar{p}\left( \alpha \right)$. To show that $\hat{p}\leq \bar{p}\left( \alpha \right)$, notice that exploiting project $L$ until news is optimal if any alternative strategy delivers weakly lower payoffs. Consider the alternative strategy that prescribes exploiting project $H$ for a time interval $\Delta$ before returning to exploiting project $L$ in the event that there is no news. This alternative strategy is superior if 
\begin{equation*} 
-\Delta r\left( R_{L}-pR_{H}\right) +\left( 1-\Delta r\right) p\lambda_{H}^{g}\Delta \alpha \frac{r}{r+\left( 1-\alpha \right) \lambda_{H}^{g}}\left( R_{H}-R_{L}\right) \leq 0. 
\end{equation*} 
Taking limits as $\Delta \rightarrow 0$ and simplifying we obtain that this requires that $p\leq \bar{p}\left( \alpha \right) $. 
\end{proof}

We now obtain expressions for the agent's payoffs, which underlie some of the results described in Section 4.2. We focus on the case of full disentanglement, $\alpha =0$, where the cutoff posterior is $\bar{p}(0)=\frac{R_{L}}{R_{H}}$, which is relevant for our analysis there. Denote by $\Omega (p)=\frac{1-p}{p}$ the odds ratio when the agent believes project $H$ is good with probability $p$. \bigskip

\begin{description} \item[Proposition A (Expected Payoffs with Full Disentanglement)] \textit{ Consider pure news settings with }$\lambda =\max \{\lambda_{H}^{g},\lambda_{H}^{b}\}$\textit{\ and }$0=\min \{\lambda_{H}^{g},\lambda_{H}^{b}\}$ \textit{. For full disentanglement, }$\alpha =0$\textit{, and posterior }$p$ \textit{\ that project }$H$\textit{\ is good,}
\begin{enumerate} \item \textit{Good news (}$\lambda =\lambda_{H}^{g}$\textit{):}
\begin{enumerate} \item \textit{If }$p\leq \bar{p}(0)$\textit{, expected payoffs are }$R_{L}+p \frac{\lambda }{r+\lambda }(R_{H}-R_{L})$\textit{;}
\item \textit{If }$p\geq \bar{p}(0)$\textit{, expected payoffs are} $ pR_{H}+(1-p)\left[ \frac{\Omega (p)}{\Omega (\bar{p}(0))}\right] ^{r/\lambda }\frac{\lambda }{r+\lambda }R_{L}$. \end{enumerate}
\item \textit{Bad news (}$\lambda =\lambda_{H}^{b}$\textit{):}
\begin{enumerate} \item \textit{If }$p\leq \bar{p}(0)$\textit{, expected payoffs are }$R_{L}+p\left[ \frac{\Omega (\bar{p}(0))}{\Omega (p)}\right] ^{r/\lambda }\frac{ \lambda }{r+\lambda }(R_{H}-R_{L})$\textit{;}
\item \textit{If }$p\geq $\textit{\ }$\bar{p}(0)$\textit{, expected payoffs are }$pR_{H}+(1-p)\frac{\lambda }{r+\lambda }R_{L}$. \end{enumerate} \end{enumerate} \end{description}

\begin{proof}[Proof of Proposition A]
The terms corresponding to parts 1.a and 2.b have already been calculated in the proof of Proposition \ref{prop1}. We now turn to parts 1.b and 2.a.

Consider good news settings and suppose $p\geq \bar{p}(0)=\frac{R_{L}}{R_{H}} $. Set $\beta $ to satisfy $p=\beta \bar{p}(1)+(1-\beta )$, so that $\beta $ is the probability such that, if the agent explores a good project $H$---generating either good news and a posterior of $1$, or no news---she will reach the posterior $\bar{p}(0)$. Let $z$ be such that $\beta =pz+(1-p)$, so that $z$ is the conditional probability that the agent reaches the posterior $\bar{p}(0)$, conditional on project $H$ being good. Simple algebra than yields that $z=\frac{\Omega (p)}{\Omega (\bar{p}(0))}$. Let $\bar{t}$ denote the exploration duration of project $H$ after which, absent news, the agent reaches precisely the posterior $\bar{p}(1)$. Since good news arrives at an exponential rate of $\lambda $,\ we can write $z=e^{-\lambda \bar{t}}$. Thus, the discount factor at time $\bar{t}$ can be written as $z^{r/\lambda } $.

Consider an auxiliary problem $\Gamma _{A}$ in which, after reaching the posterior $\bar{p}(0)$, the agent receives balanced news, with arrival rate of $\lambda = \lambda_{H}^{g}$, about project $H$ no matter which project she exploits. The optimal strategy in our setting is optimal in $\Gamma _{A}$ and, additionally, generates the same expected payoffs in both problems. Furthermore, in $\Gamma_{A}$, absent news, the agent is indifferent between exploiting project $L$ or project $H$ when reaching $\bar{p}(0)=\frac{R_{L}}{R_{H}}$. Thus, the payoffs from utilizing the optimal strategy in our setting coincide with those derived from the exploitation of project $H$ until news in $\Gamma _{A}$.

In either our problem or $\Gamma _{A}$, if the agent exploits project $H$ indefinitely, regardless of whether news arrives, she receives the expected value of project $H$, namely $pR_{H}$. Until the posterior $\bar{p}(0)$ is reached, the agent exploits project $H$ and can only learn good news about it. She therefore never switches her exploited project. The benefit of responding to news starting from $\bar{p}(0)$ is that when project $H$ is bad, which occurs with probability $(1-p)$, the arrival of news---at a time with expected discount of $\frac{\lambda }{r+\lambda }$ (see equation \ref{Kroovit})---the agent switches to project $L$ and receives $R_{L}$. Thus, the agent's expected payoff is: 
\begin{equation*} 
pR_{H}+(1-p)z^{r/\lambda }\frac{\lambda }{r+\lambda }R_{L}, 
\end{equation*} 
corresponding to the statement in part 1.b of the proposition.

Consider bad news settings and suppose $p\leq \bar{p}(0)=\frac{R_{L}}{R_{H}}$. Similar arguments to those used for good news settings imply that if we define $\tilde{z}=\frac{\Omega (\bar{p}(0))}{\Omega (p)}$, then $z^{r/\lambda}$ captures the discount factor at the time $\bar{t}$ it takes to reach $\bar{p}(0)$ when exploring project $H$ without news.

Consider an auxiliary problem $\Gamma _{A}$ as before, whereby after reaching $\bar{p}(0)$, the agent received balanced news, with arrival rate of $\lambda=\lambda_{H}^{g}$. The optimal strategy in our setting is optimal in $\Gamma _{A}$ and, additionally, generates the same expected payoffs in both problems. Until the posterior $\bar{p}(0)$ is reached, the agent exploits project $L$ and can only learn bad news about project $H$. She therefore never switches her exploited project. The benefit of responding to news starting from $\bar{p}(0)$ is that when project $H$ is good, which occurs with probability $p$, when news arrives, associated with an expected discount of $\frac{\lambda }{ r+\lambda }$, the agent switches to project $L$ and receives $R_{H}$. Thus, the agent's expected payoff is: 
\begin{equation*} 
R_{L}+pz^{r/\lambda_{H}}\frac{\lambda }{r+\lambda }(R_{H}-R_{L}), 
\end{equation*} 
which corresponds to the expression stated in part 2.a of the proposition.
\end{proof}

\bigskip

In Section 4.2, we evaluated the expected payoff benefit of disentangling exploration from exploitation. The description of payoffs when there is full entanglement, $\alpha =1$, follows from KRC's and KR's analysis. Recalling that $\bar{p}(1)=\frac{rR_{L}}{R_{H}(r+\lambda_{H})-R_{L}\lambda_{H}}$ and using the same notation as above, we have:

\begin{description} \item[Proposition B (Expected Payoffs with Full Entanglement)] \textit{ Consider pure news settings with }$\lambda =\max \{\lambda_{H}^{g},\lambda_{H}^{b}\}$\textit{\ and }$0=\min \{\lambda_{H}^{g},\lambda_{H}^{b}\}$ \textit{. For full entanglement, }$\alpha =1$\textit{, and posterior }$p$ \textit{\ that project }$H$\textit{\ is good,}

\begin{enumerate} \item \textit{Good news (}$\lambda =\lambda_{H}^{g}$\textit{):}

\begin{enumerate} \item \textit{If }$p\leq \bar{p}(1)$\textit{, expected payoffs are }$R_{L}$ \textit{;}

\item \textit{If }$p\geq \bar{p}(1)$\textit{, expected payoffs are }$pR_{H}+ \frac{1-p}{1-\bar{p}(1)}\left[ \frac{\Omega (p)}{\Omega (\bar{p}(1))}\right] ^{r/\lambda }\left( R_{L}-\bar{p}(1)R_{H}\right) $. \end{enumerate}

\item \textit{Bad news (}$\lambda =\lambda_{H}^{b}$\textit{):}

\begin{enumerate} \item \textit{If }$p\leq \bar{p}(1)$\textit{, expected payoffs are }$R_{L}$ \textit{;}

\item \textit{If }$p\geq $\textit{\ }$\bar{p}(1)$\textit{, expected payoffs are }$pR_{H}+(1-p)\frac{\lambda }{r+\lambda }R_{L}$. \end{enumerate} \end{enumerate} \end{description}

\bigskip

\subsection{Two Risky Projects: Proofs}

\begin{proof}[Proof of Proposition \ref{prop2}] Denote by $\rho _{z}=\lambda_{z}/(r+\lambda_{z})$ for $z=L,H$ the expected discount at the time at which news arrives on project $z$. Let $e_{0}=\max \left\{ {p_{L}R_{L},p_{H}R_{H}}\right\} $ be the expected payoff absent any information. Let $e_{z}$ be the expected payoff generated when the agent knows whether project $z$ is good, but has no access to information on the other project. Finally, let $e^{\ast }$ denote the expected payoff the agent receives when she has complete information on the quality of both projects.

If the agent explores project $x$ until news, and then switches to exploring
project $y\neq x$, her expected payoff is
\begin{equation*}
(1-\rho _{x})e_{0}+\rho _{x}(1-\rho _{y})e_{x}+\rho _{x}\rho _{y}e^{\ast }.
\end{equation*}
In particular, exploring project $x$ first is optimal whenever 
\begin{equation*}
(1-\rho _{x})e_{0}+\rho _{x}(1-\rho _{y})e_{x}\geq (1-\rho _{y})e_{0}+\rho
_{y}(1-\rho _{x})e_{2}.
\end{equation*}
Equivalently, 
\begin{equation*}
\rho _{x}(1-\rho _{y})(e_{x}-e_{0})\geq \rho _{y}(1-\rho _{x})(e_{y}-e_{0}),
\end{equation*}
or 
\begin{equation*}
\frac{\rho _{x}}{1-\rho _{x}}(e_{x}-e_{0})\geq \frac{\rho _{y}}{1-\rho _{y}}(e_{y}-e_{0}),
\end{equation*}
which translates to 
\begin{equation}
\lambda_{x}(e_{x}-e_{0})\geq \lambda_{y}(e_{y}-e_{0}).  \label{squirrel}
\end{equation}
If project $x$ is favorable, then $e_{0}=p_{x}R_{x}$, and $ e_{x}=p_{x}R_{x}+\left( 1-p_{x}\right) p_{y}R_{y}$. Therefore, $ e_{x}-e_{0}=\left( 1-p_{x}\right) p_{y}R_{y}$. If project $x$ is unfavorable, then $e_{0}=p_{y}R_{y}$ and $e_{x}=p_{x}\max \left( R_{x},p_{y}R_{y}\right) +\left( 1-p_{x}\right) p_{y}R_{y}$, so $ e_{x}-e_{0}=p_{x}\max \left( R_{x},p_{y}R_{y}\right) -p_{x}p_{y}R_{y}=p_{x}(R_{x}-p_{y}R_{y})^{+}=p_{x}R_{x}(1-\tilde{p}_{x})$, where $\tilde{p}_{x}=\min (p_{y}R_{y}/R_{x},1)$. By substituting into equation (\ref{squirrel}), we conclude that projects are compared via $ \lambda_{x}(1-\tilde{p}_{x})$, where $\tilde{p}_{x}=p_{x}$ when project $x$ is favorable and $\tilde{p}_{x}=\min (p_{y}R_{y}/R_{x},1)$ when project $x$ is unfavorable, as stated in the proposition.
\end{proof}

\begin{proof}[Proof of Proposition \ref{prop3}] Suppose project $x$ is favorable, so that $ p_{x}R_{x}\geq p_{y}R_{y}$. We need to show that it is optimal for the agent to either explore project $x$ for a period $T$ absent news, with $0\leq T\leq \bar{t}_{x}(p_{L},p_{H})$, after which project $y$ is explored until news is received; or to explore project $x$ until news arrives, denoted as exploring $x$ for a duration $T=\infty $. Whenever the agent receives news on one project, but not the other, she reverts to exploring the uncertain project. Once the agent learns the realization of both projects, the exploration strategy has no payoff impacts. For simplicity, we assume the agent reverts to exploring project $x$ in that case. We denote by $\sigma _{T}$ the strategy induced by each such $T\in \lbrack 0,\bar{t} _{H}(p_{L},p_{H})]\cup \{\infty \}$.

Given the original decision problem $\Gamma$, consider an auxiliary problem $\Gamma_{A}$ with the following modified news process:

\begin{enumerate} 
\item If the agent explores project $y$, she receives both good and bad news at a rate of $\lambda_{y}^{g}$.
\item If the agent explores project $x$ \textit{and} by that moment she has already explored project $x$ for a period of at least $\bar{t}_{x}(p_{L},p_{H})$, she receives both good and bad news at a rate of $\lambda_{x}^{g}$.
\item If the agent explores project $x$ \textit{and} by that moment she has explored project $x$ for a period smaller than $\bar{t}_{x}(p_{L},p_{H})$, she receives good news at a rate of $\lambda_{x}^{g}$ and bad news at a rate of $\lambda_{x}^{b}$. 
\end{enumerate}

Under any exploration strategy, and at any point in time, the agent is at least as well informed in $\Gamma _{A}$ as in $\Gamma $. In particular, the optimal expected payoff that can be achieved in $\Gamma_{A}$ is weakly higher than the optimal expected payoff that can achieved in $\Gamma $.

\vspace{2.5mm}

\item[Claim A1] \textit{For any }$T\in \lbrack 0,\bar{t}_{x}(p_{L},p_{H})]\cup \{\infty \}$\textit{, the strategy }$\sigma _{T}$\textit{\ generates the same expected payoff in }$\Gamma _{A}$\textit{\ as it does in }$\Gamma $ \textit{.}

\item[Proof of Claim A.1] For $T\leq \bar{t}_{x}(p_{L},p_{H})$, the agent receives information at the same arrival rate in both $\Gamma $ and $\Gamma _{A}$ during the initial duration of $T$. If news arrives during that period, the resulting optimal exploitation is identical in both problems: if good news arrives, exploit project $x$ indefinitely if $x=H$ or until good news arrives from project $y$ if $x=L$, and if bad news arrives from project $x$, then exploit project $y$ indefinitely.   Absent news, project $x$ remains favorable when the agent switches to exploring project $y$. Thus, from then on, only good news on project $y$ alters her exploitation. Since the arrival rate of good new on project $y$ is the same in $\Gamma $ and $\Gamma _{A}$, the resulting expected payoffs coincide as well.

Suppose now that $T=\infty $, so that the agent explores project $x$ until receiving news. Until time $\bar{t}_{x}(p_{L},p_{H})$, news arrives at the same rate in both $\Gamma $ and $\Gamma _{A}$. Absent news, at time $\bar{t}_{x}(p_{L},p_{H})$, the agent is indifferent between the two projects: they are both favorable. At any $t>$ $\bar{t}_{x}(p_{L},p_{H})$, absent news, it is optimal to exploit project $y$ in both $\Gamma $ and $\Gamma _{A}$. Only good news on project $x$ then alters exploitation, and good news arrives at the same rate in $\Gamma $ and $\Gamma _{A}$. Therefore, the resulting expected payoffs coincide.

\vspace{2.5mm}

\item[Claim A.2] \textit{There exists }$T\in \lbrack 0,\bar{t} _{x}(p_{L},p_{H})]\cup \{\infty \}$\textit{\ such that }$\sigma _{T}$\textit{ \ is optimal in }$\Gamma _{A}$.

\item[Proof of Claim A.2] In $\Gamma _{A}$, if the agent explores project $y$ and sees no news, her belief about the quality of project $y$ does not change. Therefore, by dynamic-programming principles, if it is optimal for the agent to explore project $y$ at any point then, absent news, it is also optimal to explore project $y$ at any later point. Similarly, if the agent has explored project $x$ for a period of at least $\bar{t}_{x}(p_{L},p_{H})$, continuing to explore project $x$ until news is optimal. The conclusion follows.

\vspace{2.5mm}

Claims A.1 and A.2 illustrate the optimality of the class of strategies specified in the proposition. We now turn to showing that in settings with pure good news on at least one project, exploration switches only upon receiving news.

\vspace{2.5mm}

\item[Claim A.3] \textit{If $\lambda_{x}^{b}=0$, there exists an optimal strategy in $\Gamma $ with $T=0$ or $T=\infty $.}

\item[Proof of Claim A.3] Suppose Alex explores project $y$ from the start, i.e., Alex uses the strategy $\sigma _{0}$. Bailey, facing the same decision problem, uses $\sigma _{T}$ with $0<T\leq \bar{t}_{x}(p_{L},p_{H})$. We claim that Alex has a higher expected payoff than Bailey.

Consider Alexis and Baylor, who face a coupled problem. Baylor, like Bailey, explores project $x$ for a period of $T$ or until receiving news. Denote by $ \omega $ the random time when Baylor either receives news on project $x$ or a period of $T$ has transpired (so that $\omega $ is the minimum between $T$ and the arrival time of news on project $x$, which is distributed exponentially with arrival rates of $\lambda_{x}^{b}=0$ and $\lambda_{x}^{g}$ ). Like Bailey, after time $\omega $, Baylor switches to exploring project $ y $. Unlike Bailey, at any time $t\geq \omega $, Baylor receives the news Alexis has received at time $t-\omega $ on project $y$. Alexis, like Alex, explores project $y$ until news. Let $\tau $ be the random variable that represents the first arrival of news on project $y$ for Alex (distributed exponentially with parameters $\lambda_{y}^{b}$ and $\lambda_{y}^{g}$). At any time $t\in \lbrack \tau ,\tau +\omega ]$, Alexis receives the news Baylor has received on project $x$ at time $t-\tau $, after which Alexis receives news independently on project $x$. Thus, Alexis' and Baylor's information is coupled. Since Alex and Bailey's news arrivals are independent and identical, Alexis receives the same expected payoff as Alex and Baylor receives the same expected payoffs as Bailey. We now show that Alexis receives a weakly higher expected payoff than Baylor.

Conditional on $\tau $, at any moment $t$ such that $0\leq t\leq \omega +\tau $, Baylor does not learn whether project $y$ is good or bad. Since $ \lambda_{x}^{b}=0$, Baylor can only receive good news or no news about project $x$ until such time $t$. Since $T\leq \bar{t}_{x}(p_{L},p_{H})$, in either case, Baylor continues exploiting project $x$. Alexis, however, exploits project $x$ until time $\tau $, when a switch to project $y$ may be optimal when news about project $y$ is good. Therefore, conditional on $ \omega $ and $\tau $, up to time $\min \{\omega ,\tau \}$, Alexis' and Baylor's expected payoffs coincide, whereas over the period between $\min \{\omega ,\tau \}$ and $\omega +\tau $, Alexis' expected payoff is weakly higher than Baylor's. At any moment $t$ such that $t>\omega +\tau $, both Alexis and Baylor know whether project $y$ is good or bad and have explored project $x$ for a period $t-\tau $, receiving the same information ex-ante.\footnote{Recall that we assumed the agent explores the ex-ante favorable project $x$ after receiving news on project $y$, even when having received news on project $x$ as well.} Therefore, at moments $t$ such that $t>\omega +\tau $, Alexis' and Baylor's expected payoffs are the same. Therefore, Alexis' expected payoff is weakly higher then Baylor's, as required.
\end{proof}

\begin{proof}[Proof of Claim \ref{claim1}] We show that the unique optimal strategy is the strategy described in Proposition \ref{prop3} and that the optimal switching point of exploration is generically unique. It suffices to prove this for an auxiliary problem as in the proof of Proposition \ref{prop3}, in which balanced news from the unfavorable arm $y$ arrives at rate $\lambda^g_y$.

Fix all parameters except $p_x$. Consider the difference in the agent's expected payoffs between a strategy that explores project $x$ for a small duration $\Delta>0$ and then switches to exploring project $y$ forever and a strategy that explores project $y$
forever. This difference is given by 
\begin{multline}\label{fo-good}
\Delta*\left((1-p_x)\lambda_x^b*\frac{r}{r+\lambda^g_y}p_yR_y-rp_y\frac{\lambda^g_y}{r+\lambda^g_y}(R_y-p_xR_x)\right)+o(\Delta)
=\\\Delta*\frac{r}{r+\lambda^g_y}p_yR_y D(p_x)+o(\Delta),
\end{multline}
with $D(p_x)=\lambda_x^b(1-p_x) - \lambda_y^g(1-p_xR_x/R_y)$.
The first term  is the benefit from the fact that the agent learns that project $x$ is bad and then, until receiving news from project $y$, exploits project $y$ (instead of getting a payoff of $0$ from project $x$). The second term corresponds to the costs incurred when project $y$ is good and there is a duration $\Delta$ in which the agent would have exploited project $y$, but exploits project $x$ instead. 

In the auxiliary problem, and as long as project $x$ remains favorable, an optimal strategy must explore project $x$ if $D(p_x)>0$ and cannot switch from exploring project $x$ to exploring project $y$ if $D(p_x)<0$. Thus, an optimal strategy entails an exploration switch from project $x$ to project $y$ at the point in which $D(p_x)=0$.  

Now, $D(p_x)$ is linear in $p_x$ with slope $-\lambda_x^b+\lambda_y^g R_x/R_y$. As long as $\lambda_x^b R_y \neq \lambda_y^g R_x$, which is generically the case, $D(p_x)$ has a non-trivial slope and there is only one point in which exploration can switch from project $x$ to project $y$. 
\end{proof}

\begin{proof}[Proof of Claim \ref{claim2}] If project $L$ is explored, only good news yields a switch of the exploited project. If project $H$ is explored, absent news, the agent switches her exploited project after $\bar{t}_{H}(p_{L},p_{H})$ has passed, when she is indifferent between the expected payoffs of both projects. We now characterize $\bar{t} _{H}(p_{L},p_{H})$, where we drop the arguments when there is no risk of confusion.

By definition, after a duration $\bar{t}_{H}$ of exploring project $H$, the agent's posterior that project $H$ is good declines to $qp_{H}$, where $q=\frac{p_{L}R_{L}}{p_{H}R_{H}}\in (0,1)$. Certainly, if the agent receives good news on project $H$ before reaching indifference, the corresponding posterior jumps to $1$. The conditional probability that the agent reaches indifference when exploring project $H$, conditional on project $H$ being good, is therefore $\frac{q(1-p_{H})}{1-qp_{H}}$.\footnote{The arguments are reminiscent of those used in the proof of Proposition A. Set $\beta $ to satisfy $p_{H}=\beta qp_{H}+(1-\beta )$, so that $\beta $ is the probability such that, if the agent explores project $H$, she will reach a time at which she is indifferent between the projects. Let $z$ be such that $\beta =p_{H}z+(1-p_{H})$, so that $z$ is the conditional probability that the agent reaches indifference, conditional on project $H$ being good. Simple algebra yields the specified formula.} The exponential distribution of news then yields: 
\begin{equation*} 
e^{-\lambda_{H}^{g}\bar{t}_{H}}=\frac{q(1-p_{H})}{1-qp_{H}}. 
\end{equation*} 
The discount at the indifference time $\bar{t}_{H}$ is given by $w=e^{-r\bar{ t}_{H}}$. As before, let $\rho_{z}=\lambda_{z}^{g}/(r+\lambda_{z}^{g})$, $ z=L,H$, denote the expected discount at the time $\bar{t}_{H}$ at which news first arrives when the arrival rate is $\lambda_{z}^{g}$.

Suppose the agent explores project $H$ indefinitely. As argued in the proof of Proposition \ref{prop3}, her payoff coincides with the payoff of an agent who, after time $\bar{t}_{H}$, sees all news from project $H$---good or bad, at a (balanced) rate $\lambda_{H}^{g}$. So, a-priori, the agent expects to receive $e_{0}=p_{H}R_{H}$ up to a time that is exponentially distributed with parameter $\lambda_{H}^{g}$ beyond the indifference time $\bar{t}_{H}$ . After that time, she receives $e_{H}=p_{H}R_{H}+(1-p_{H})p_{L}R_{L}$. The agent's expected payoff is therefore: 
\begin{equation*} 
(1-w\rho _{H})e_{0}+w\rho _{H}e_{H}=e_{0}+w\rho _{H}(e_{H}-e_{0}).
\end{equation*} 
Now, suppose the agent explores project $L$ instead. Define, analogously, $ e_{L}=p_{L}R_{L}+p_{H}R_{H}(1-p_{L})$ to be the expected value from exploring project $L$ upon indifference.

As shown in the proof of Proposition \ref{prop3}, the agent's expected payoff is the same as in the balanced news setting, and equals
\begin{equation*} 
(1-\rho _{L})e_{0}+\rho _{L}(1-\rho _{H})e_{L}+\rho _{L}\rho _{H}e_{H}=e_{0}+\rho _{L}(1-\rho _{H})(e_{L}-e_{0})+\rho _{L}\rho _{H}(e_{H}-e_{0}). 
\end{equation*} 
Thus, it is optimal to explore project $H$ if and only if \begin{equation*} 
\rho _{H}(w-\rho _{L})(e_{H}-e_{0})\geq \rho _{L}(1-\rho _{H})(e_{L}-e_{0}). 
\end{equation*} 
The statement of the claim then follows.
\end{proof}

\begin{proof}[Proof of Proposition \ref{prop4}] The proof follows several claims:

\vspace{2.5mm}

\item[Claim B.1] \textit{There exists an optimal policy with the property that, if it is optimal to explores project $H$ at some point when it is favorable, then, from that point on, it is optimal to explore project $H$ until news is received.}

\item[Proof of Claim B.1] Consider an auxiliary problem $\Gamma _{A}$ that coincides with the original problem $\Gamma$ with the following modification: if the agent explores project $H$ and project $H$ is currently weakly favorable, the agent receives both good news and bad news on project $H$ at a rate of $\lambda_{H}^{b}$. In particular, the agent has more information in $\Gamma_{A}$ than in $\Gamma$.

Any strategy described in the statement of the proposition generates the same payoff in $\Gamma _{A}$ as it does in $\Gamma $. Indeed, if project $H$ is favorable, and the agent explores it, then in both $\Gamma $ and $\Gamma _{A}$, project $H$ would remain favorable as long as no bad news arrive.

It suffices to show that, under the optimal strategy in $\Gamma _{A}$, once the agent starts exploring project $H$ , she continues doing so until receiving news. Indeed, if the agent explores project $H$ in the auxiliary problem when project $H$ is currently favorable, then the state variable---her posterior---does not change. By dynamic programming principles, it must be optimal to continue exploring project $H$ until news arrives.

\vspace{2.5mm}

\item[Claim B.2] \textit{If project $L$ is favorable at the outset, then, at any point in which project $H$ becomes strictly favorable, it is optimal to explore project $H$ until receiving news.}

\item[Proof of Claim B.2] 
Project $H$ becomes strictly favorable when one of the following occurs. First, upon arrival of bad news about project $L$ or good news about project $H$, project $H$ becomes favorable and exploring either project is optimal. The second option is that the agent explores project $H$. Since $\lambda_{H}^{b}>\lambda_{H}^{g}$, over time, the agent becomes more optimistic about project $H$. In this case, by Claim B.1, the agent should continue exploring project $H$ .

\vspace{2.5mm}

For the next step of the proof, consider a balanced news setting with arrival rates $\lambda_{H}^{g}$ for project $H$ and $\lambda_{L}^{b}$ for project $L$ that starts with prior probabilities $p_{H}$ and $p_{L}$ such that $L$ is favorable. Let $\hat{p}_{L}$ be such that \ the agent explores project $H$ if $p_{L}\geq \hat{p}_{L}$, holding all other parameters fixed. By Proposition \ref{prop2}, 
\begin{equation} \lambda_{H}^{g}(1-\tilde{p}_{H})=\lambda_{L}^{b}(1-\hat{p}_{L})\text{,} \label{phat} 
\end{equation} 
where $\tilde{p}_{H}=p_{L}R_{L}/R_{H}$. 

\vspace{2.5mm} 

\item[Claim B.3] \textit{If project $L$ is favorable and $p_{L}>\hat{p}_{L}$, then it is optimal to explore project $H$ until news.}

\item[Proof of Claim B.3] Consider an auxiliary problem $\Gamma _{B}$, a modification of the original problem $\Gamma$ in which exploring project $L$ generates balanced news at a rate of $\lambda_{L}^{b} $. The agent is weakly better off in $\Gamma _{B}$ relative to the original problem $\Gamma $ since she has access to information that arrives at higer rates. Furthermore, exploring project $H$ until news generates the same payoff in $\Gamma _{B}$ as it does in $\Gamma $: news about project $H$ arrives at the same rate in both problems and, in both, exploiting project $H$ (or project $L$) forever once project $H$ is observed to be good (or bad) maximizes expected payoffs.

Suppose that project $L$ is favorable and $p_{L}>\hat{p}_{L}$. We show that, in $ \Gamma _{B}$, the agent optimally explores project $H$ until news. Assume, by way of contradiction, that it is optimal to explore project $L$ in $ \Gamma _{B}$.\footnote{ Since news is balanced on project $L$, if it is optimal to explore project $ L $ at any posterior, it is optimal to continue exploring project $L$ as long as news does not arrive.} Absent news, exploring project $L$ does not alter the agent's beliefs about the projects' quality and, therefore, it must be optimal to explore project $L$ until news. Consider a deviation to first exploring project $H$ for a short interval $\Delta $ and then exploring project $L$ until news, where $\Delta $ is sufficiently small so that, absent news during the time period $\Delta $, project $L$ remains favorable. We claim that this deviation improves payoffs. Suppose Alex plays the candidate strategy---exploring project $L$ until news---and Bailey follows the deviation.

Let $\tau _{L}$ and $\tau _{H}$ denote the random variables corresponding to the first arrival time of news on project $L$ and project $H$, respectively, where arrival rates are those specified in the auxiliary problem $\Gamma _{B}$. Both Alex and Bailey receive news on project $z=L,H$ after exploring project $z$ for a duration $\tau _{z}$. Thus, Alex's and Bailey's information is coupled. Furthermore, Alex's and Bailey's payoffs from the suggested strategies are the same as before.

The difference between Bailey's and Alex's payoffs is then:
\begin{equation*} 
p_{H}\lambda_{H}^{g}\Delta \frac{r}{r+\lambda_{L}^{b}} (R_{H}-p_{L}R_{L})-(1-p_{L})\frac{\lambda_{L}^{b}}{r+\lambda_{L}^{b}} r\Delta p_{H}R_{H}+o(\Delta). 
\end{equation*} 
The first term corresponds to the case in which Bailey receives good news on project $H$ in the initial duration of $\Delta $ (occurring with probability $ p_{H}\lambda_{H}^{g}\Delta $), while Alex is delayed in learning about project $H$ until receiving news on project $L$ at time $\tau _{L}$ (occurring at a rate of $\lambda_{L}^{b}$ whether project $L$ is good or bad). The expected discounted weight of that duration is $1-\frac{\lambda_{L}^{b}}{ r+\lambda_{L}^{b}}=\frac{r}{r+\lambda_{L}^{b}}$ (see equation (\ref{Kroovit})). The second term corresponds to project $L$ being bad. In that case, conditional on not receiving news in the first period of $\Delta $ (occurring with probability $1-\lambda_{H}^{g}\Delta $), Bailey would be delayed by $\Delta $ relative to Alex in learning that project $L$ is bad. Observing that project $L$ is bad would lead either agent to exploit project $H$, which generates an expected payoff of $p_{H}R_{H}\Delta$ (up to $o(\Delta )$ due to updating on project $H$ during the initial period of $\Delta $). The relevant expected discount at $\tau _{L}$, when Alex learns that project $L$ is bad, is $\frac{\lambda_{L}^{b}}{r+\lambda_{L}^{b}}$. 

Reorganizing terms implies that the payoff difference is: \begin{equation*} 
\Delta \frac{r}{r+\lambda_{L}^{b}}p_{H}R_{H}\left( \lambda_{H}^{g}(1- \tilde{p}_{H})-\lambda_{L}^{b}(1-p_{L})\right) +o(\Delta )>0,
\end{equation*} 
where the inequality follows from our assumption that $p_{L}>\hat{p}_{L}$. The conclusion of Claim B.3 then follows using Claim B.1.

\vspace{2.5mm}

\item[Claim B.4] \textit{If project $L$ is favorable and $p_{L}\leq \hat{p}_{L}$, it is optimal to explore project $L$ for some period, and then explore project $H$ until news.}

\item[Proof of Claim B.4] Consider an optimal strategy, and let $T$ be the first time such that, according to this strategy, if no news arrives up to time $T$, either project $H$ becomes favorable or the posterior that project $L$ is good reaches $\hat{p}_{L}$. By Claims B.2 and and B.3, if no news arrives by time $T$, it is optimal to explore project $H$. We claim that, before time $T$, it is optimal to explore project $L$ for some period and then switch to exploring project $H$.

Suppose, toward a contradiction, that the claim is violated. Then, there must be a sufficiently small $\Delta $, a fraction $\beta >0$, and times $ t^{\prime }<t^{\prime \prime }<T$ with $t^{\prime }-\Delta >0$ and $ t^{\prime \prime }+\Delta <T,$ such that the agent optimally explores project $H$ for an amount of time $\beta \Delta $ in $I^{\prime }=[t^{\prime }-\Delta ,t^{\prime }]$ and explores project $L$ for an amount of time $ \beta \Delta $ in $I^{\prime \prime }=[t^{\prime \prime },t^{\prime \prime }+\Delta ]$.

We now show that swapping the order of these $\beta \Delta $ exploration resources between the intervals $I^{\prime }$ and $I^{\prime \prime }$ improves the agent's expected payoff. Indeed, suppose Alex plays the candidate strategy and Bailey performs the swap, and their news are coupled as follows:

\begin{enumerate} 
\item All news coming from exploration that was not interchanged, which we call \emph{regular news}, are the same for Alex and Bailey.
\item The \emph{additional news on project $L$} that Bailey receives from the additional $\beta \Delta $ exploration during $I^{\prime }$ is received by Alex during $I^{\prime \prime }$
\item The \emph{additional news on project $H$} that Alex receives from the additional $\beta \Delta $ exploration during $I^{\prime }$ is the news received by Bailey during $I^{\prime \prime }$. \end{enumerate}

We need to show that Bailey's payoff is higher than Alex's. We will, in fact, show that this is the case even if Bailey does not play optimally: we assume that if Bailey receives additional good news about project $L$, Bailey ignores this news and switches to exploring project $H$ only when either regular good news arrives about project $L$ or when Alex receives the additional good news about project $L$ (in which case Alex also switches to only exploring project $H$).

Until time $T$, Alex and Bailey both exploit project $L$ unless they received bad news from project $L$ or good news from project $H$. They gain different payoffs at time $ t\in \lbrack t^{\prime },t^{\prime \prime }]$ only if they receive no regular news up to time $t$ and either
\begin{enumerate} 
\item bad news on project $L$ is received by Bailey over $I^{\prime }$, in which case Bailey exploits project $H$, while Alex exploits project $L$; or
\item good news on project $H$ is received only by Alex over $I^{\prime }$, in which case Alex exploits project $H$, while Bailey exploits project $L$. 
\end{enumerate}

Therefore, the difference in expected payoffs is 
\begin{equation*} 
\beta \Delta \int_{t^{\prime }}^{t^{\prime \prime }}re^{-rt}\rho (t)\left[ \lambda_{L}^{b}(1-p_{L}(t))p_{H}(t)R_{H}-\lambda_{H}^{g}p_{H}(t)(R_{H}-p_{L}(t)R_{L}\right] ~\text{d}t+O(\Delta), 
\end{equation*} 
where $\rho (t)$ is the probability that there were no regular news until time $t$; the probabilities $p_{H}(t)$ and $p_{L}(t)$ are, respectively, the conditional probabilities that projects $H$ and $L$ are good given this event; and $ \tilde{p}_{H}(t)=p_{L}(t)R_{L}/R_{H}$. Rearranging terms, this payoff difference equals: \begin{equation*} \beta \Delta \int_{t^{\prime }}^{t^{\prime \prime }}re^{-rt}\rho (t)p_{H}(t)R_{H}\left[ \lambda_{L}^{b}(1-p_{L}(t))-\lambda_{H}^{g}(1- \tilde{p}_{H}(t)\right] ~\text{d}t+o(\Delta)>0\text{,} \end{equation*} where the inequality follows from the fact that $p_{L}(t)<\hat{p}_{L}$ for every $t<t^{\prime \prime }$.

\vspace{2.5mm}

\item[Claim B.5] \textit{If project $H$ is favorable, it is optimal to explore project $L$ for some period, and then explore project $H$ until news.}

\item[Proof of Claim B.5] Suppose $R_{L}>p_{H}R_{H}\geq p_{L}R_{L}$. From Claim B.1, once the agent starts exploring project $H$, it is optimal to do so until news. Towards a contradiction, suppose the agent explores project $ L $ until news. Absent news, at any time $t>\bar{t}_{L}(p_{L},p_{H})$, project $L$ becomes favorable. Claims B.3 and B.4 then lead to a contradiction.

If $p_{H}R_{H}\geq R_{L}$, news on project $L$ cannot generate a switch in the agent's exploited project and exploring project $L$ indefinitely is dominated. The claim then follows directly from Claim B.1.

\vspace{2.5mm}

The proposition follows from Claims B.3, B.4, and B.5.
\end{proof}

\newpage 

\singlespacing \bibliographystyle{chicago} 
\bibliography{ExplorationExploitation} 

\begin{thebibliography}{}

\bibitem[\protect\citeauthoryear{Audibert, Bubeck, and Munos}{Audibert
  et~al.}{2010}]{audibert2010best}
Audibert, J.-Y., S.~Bubeck, and R.~Munos (2010).
\newblock Best arm identification in multi-armed bandits.
\newblock In {\em COLT}, pp.\  41--53.

\bibitem[\protect\citeauthoryear{Bardhi, Guo, and Strulovici}{Bardhi
  et~al.}{2020}]{bardhi2020early}
Bardhi, A., Y.~Guo, and B.~Strulovici (2020).
\newblock Early-career discrimination: Spiraling or self-correcting?
\newblock {\em mimeo\/}.

\bibitem[\protect\citeauthoryear{Bergemann and Hege}{Bergemann and
  Hege}{1998}]{bergemann1998venture}
Bergemann, D. and U.~Hege (1998).
\newblock Venture capital financing, moral hazard, and learning.
\newblock {\em Journal of Banking \& Finance\/}~{\em 22\/}(6-8), 703--735.

\bibitem[\protect\citeauthoryear{Bergemann and Valimaki}{Bergemann and
  Valimaki}{2006}]{bergemann2006bandit}
Bergemann, D. and J.~Valimaki (2006).
\newblock Bandit problems.

\bibitem[\protect\citeauthoryear{Bolton and Harris}{Bolton and
  Harris}{1999}]{bolton1999strategic}
Bolton, P. and C.~Harris (1999).
\newblock Strategic experimentation.
\newblock {\em Econometrica\/}~{\em 67\/}(2), 349--374.

\bibitem[\protect\citeauthoryear{Bubeck, Munos, and Stoltz}{Bubeck
  et~al.}{2011}]{bubeck2011pure}
Bubeck, S., R.~Munos, and G.~Stoltz (2011).
\newblock Pure exploration in finitely-armed and continuous-armed bandits.
\newblock {\em Theoretical Computer Science\/}~{\em 412\/}(19), 1832--1852.

\bibitem[\protect\citeauthoryear{Che and H{\"o}rner}{Che and
  H{\"o}rner}{2018}]{che2018recommender}
Che, Y.-K. and J.~H{\"o}rner (2018).
\newblock Recommender systems as mechanisms for social learning.
\newblock {\em The Quarterly Journal of Economics\/}~{\em 133\/}(2), 871--925.

\bibitem[\protect\citeauthoryear{Che and Mierendorff}{Che and
  Mierendorff}{2019}]{che2019optimal}
Che, Y.-K. and K.~Mierendorff (2019).
\newblock Optimal dynamic allocation of attention.
\newblock {\em American Economic Review\/}~{\em 109\/}(8), 2993--3029.

\bibitem[\protect\citeauthoryear{Crawford and Shum}{Crawford and
  Shum}{2005}]{crawford2005uncertainty}
Crawford, G.~S. and M.~Shum (2005).
\newblock Uncertainty and learning in pharmaceutical demand.
\newblock {\em Econometrica\/}~{\em 73\/}(4), 1137--1173.

\bibitem[\protect\citeauthoryear{Currie and MacLeod}{Currie and
  MacLeod}{2020}]{currie2020understanding}
Currie, J.~M. and W.~B. MacLeod (2020).
\newblock Understanding doctor decision making: The case of depression
  treatment.
\newblock {\em Econometrica\/}~{\em 88\/}(3), 847--878.

\bibitem[\protect\citeauthoryear{Damiano, Li, and Suen}{Damiano
  et~al.}{2020}]{damiano2020learning}
Damiano, E., H.~Li, and W.~Suen (2020).
\newblock Learning while experimenting.
\newblock {\em The Economic Journal\/}~{\em 130\/}(625), 65--92.

\bibitem[\protect\citeauthoryear{Dickstein et~al.}{Dickstein
  et~al.}{2021}]{dickstein2021efficient}
Dickstein, M.~J. et~al. (2021).
\newblock {\em Efficient provision of experience goods: Evidence from
  antidepressant choice}.

\bibitem[\protect\citeauthoryear{Eliaz, Fershtman, and Frug}{Eliaz
  et~al.}{2024}]{eliaz2022optimal}
Eliaz, K., D.~Fershtman, and A.~Frug (2024).
\newblock On optimal scheduling.
\newblock {\em American Economic Journal: Microeconomics\/}.

\bibitem[\protect\citeauthoryear{Gittins, Glazebrook, and Weber}{Gittins
  et~al.}{2011}]{gittins2011bandits}
Gittins, J., K.~Glazebrook, and R.~Weber (2011).
\newblock {\em Multi-armed bandit allocation indices}.
\newblock John Wiley \& Sons.

\bibitem[\protect\citeauthoryear{Gittins}{Gittins}{1979}]{gittins1979bandit}
Gittins, J.~C. (1979).
\newblock Bandit processes and dynamic allocation indices.
\newblock {\em Journal of the Royal Statistical Society Series B: Statistical
  Methodology\/}~{\em 41\/}(2), 148--164.

\bibitem[\protect\citeauthoryear{Gittins and Jones}{Gittins and
  Jones}{1979}]{gittins1979dynamic}
Gittins, J.~C. and D.~M. Jones (1979).
\newblock A dynamic allocation index for the discounted multiarmed bandit
  problem.
\newblock {\em Biometrika\/}~{\em 66\/}(3), 561--565.

\bibitem[\protect\citeauthoryear{Guo}{Guo}{2016}]{guo2016dynamic}
Guo, Y. (2016).
\newblock Dynamic delegation of experimentation.
\newblock {\em American Economic Review\/}~{\em 106\/}(8), 1969--2008.

\bibitem[\protect\citeauthoryear{H{\"o}rner and Samuelson}{H{\"o}rner and
  Samuelson}{2013}]{horner2013incentives}
H{\"o}rner, J. and L.~Samuelson (2013).
\newblock Incentives for experimenting agents.
\newblock {\em The RAND Journal of Economics\/}~{\em 44\/}(4), 632--663.

\bibitem[\protect\citeauthoryear{Jovanovic}{Jovanovic}{1979}]{jovanovic1979job}
Jovanovic, B. (1979).
\newblock Job matching and the theory of turnover.
\newblock {\em Journal of Political Economy\/}~{\em 87\/}(5, Part 1), 972--990.

\bibitem[\protect\citeauthoryear{Keller and Rady}{Keller and
  Rady}{2010}]{keller2010strategic}
Keller, G. and S.~Rady (2010).
\newblock Strategic experimentation with poisson bandits.
\newblock {\em Theoretical Economics\/}~{\em 5\/}(2), 275--311.

\bibitem[\protect\citeauthoryear{Keller, Rady, and Cripps}{Keller
  et~al.}{2005}]{keller2005strategic}
Keller, G., S.~Rady, and M.~Cripps (2005).
\newblock Strategic experimentation with exponential bandits.
\newblock {\em Econometrica\/}~{\em 73\/}(1), 39--68.

\bibitem[\protect\citeauthoryear{Liang and Mu}{Liang and
  Mu}{2020}]{liang2020complementary}
Liang, A. and X.~Mu (2020).
\newblock Complementary information and learning traps.
\newblock {\em The Quarterly Journal of Economics\/}~{\em 135\/}(1), 389--448.

\bibitem[\protect\citeauthoryear{Liang, Mu, and Syrgkanis}{Liang
  et~al.}{2022}]{liang2022dynamically}
Liang, A., X.~Mu, and V.~Syrgkanis (2022).
\newblock Dynamically aggregating diverse information.
\newblock {\em Econometrica\/}~{\em 90\/}(1), 47--80.

\bibitem[\protect\citeauthoryear{Ma{\'c}kowiak, Mat{\v{e}}jka, and
  Wiederholt}{Ma{\'c}kowiak et~al.}{2023}]{mackowiak2023rational}
Ma{\'c}kowiak, B., F.~Mat{\v{e}}jka, and M.~Wiederholt (2023).
\newblock Rational inattention: A review.
\newblock {\em Journal of Economic Literature\/}~{\em 61\/}(1), 226--273.

\bibitem[\protect\citeauthoryear{Miller}{Miller}{1984}]{miller1984job}
Miller, R.~A. (1984).
\newblock Job matching and occupational choice.
\newblock {\em Journal of Political Economy\/}~{\em 92\/}(6), 1086--1120.

\bibitem[\protect\citeauthoryear{Robbins}{Robbins}{1952}]{robbins1952some}
Robbins, H. (1952).
\newblock Some aspects of the sequential design of experiments.

\bibitem[\protect\citeauthoryear{Rothschild}{Rothschild}{1974}]{rothschild1974two}
Rothschild, M. (1974).
\newblock A two-armed bandit theory of market pricing.
\newblock {\em Journal of Economic Theory\/}~{\em 9\/}(2), 185--202.

\bibitem[\protect\citeauthoryear{Sims}{Sims}{2003}]{sims2003implications}
Sims, C.~A. (2003).
\newblock Implications of rational inattention.
\newblock {\em Journal of monetary Economics\/}~{\em 50\/}(3), 665--690.

\bibitem[\protect\citeauthoryear{Strulovici}{Strulovici}{2010}]{strulovici2010learning}
Strulovici, B. (2010).
\newblock Learning while voting: Determinants of collective experimentation.
\newblock {\em Econometrica\/}~{\em 78\/}(3), 933--971.

\bibitem[\protect\citeauthoryear{Thompson}{Thompson}{1933}]{thompson1933likelihood}
Thompson, W.~R. (1933).
\newblock On the likelihood that one unknown probability exceeds another in
  view of the evidence of two samples.
\newblock {\em Biometrika\/}~{\em 25\/}(3-4), 285--294.

\bibitem[\protect\citeauthoryear{Wald}{Wald}{1947}]{wald1947foundations}
Wald, A. (1947).
\newblock Foundations of a general theory of sequential decision functions.
\newblock {\em Econometrica\/}, 279--313.

\bibitem[\protect\citeauthoryear{Zhuo}{Zhuo}{2023}]{zhuo2023exploit}
Zhuo, R. (2023).
\newblock Exploit or explore? an empirical study of resource allocation in
  research labs.
\newblock {\em mimeo\/}.

\end{thebibliography}

\end{document}